\documentclass[acmsmall,screen,nonacm]{acmart}

\AtEndPreamble{%
  \theoremstyle{acmdefinition}
  \newtheorem{remark}[theorem]{Remark}
  
  \theoremstyle{acmtheorem}
  \newtheorem*{claim*}{Claim}}

\usepackage{algorithm}
\usepackage[noend]{algpseudocode}

\usepackage{amsmath,mathtools,stmaryrd,color,alltt}
\usepackage{bbm}
\usepackage{etex}
\usepackage[nameinlink]{cleveref}
\usepackage{wrapfig}

\crefname{theorem}{Theorem}{Theorems}
\crefname{lemma}{Lemma}{Lemmas}
\crefname{corollary}{Corollary}{Corollaries}
\crefname{proposition}{Proposition}{Propositions}
\crefname{definition}{Definition}{Definitions}
\crefname{remark}{Remark}{Remarks}
\crefname{notation}{Notation}{Notations}

\crefname{section}{Section}{Sections}
\crefname{algorithm}{Algorithm}{Algorithms}

\newif\ifdraft\draftfalse
\newif\ifappendix
\appendixtrue
\input{local.sty}
\input{symbols.sty}

\begin{document}

\title{A Primal-Dual Perspective on Program Verification Algorithms (Extended Version)}

\author{Takeshi Tsukada}
\affiliation{%
  \institution{Chiba University}
  \city{Chiba}
  \country{Japan}}
\email{t.tsukada@acm.org}
\orcid{0000-0002-2824-8708}

\author{Hiroshi Unno}
\affiliation{%
  \institution{Tohoku University}
  \city{Sendai}
  \country{Japan}}
\email{hiroshi.unno@acm.org}
\orcid{0000-0002-4225-8195}

\author{Oded Padon}
\affiliation{%
  \institution{Weizmann Institute of Science}
  \city{Rehovot}
  \country{Israel}
}
\email{oded.padon@weizmann.ac.il}
\orcid{0009-0006-4209-1635}

\author{Sharon Shoham}
\affiliation{%
  \institution{Tel Aviv University}
  \city{Tel Aviv}
  \country{Israel}
}
\email{sharon.shoham@gmail.com}
\orcid{0000-0002-7226-3526}

\begin{abstract}
  Many algorithms in verification and automated reasoning leverage some form of duality between proofs and refutations or counterexamples. In most cases, duality is only used as an intuition that helps in understanding the algorithms and is not formalized. In other cases, duality is used explicitly, but in a specially tailored way that does not generalize to other problems.

In this paper we propose a unified primal-dual framework for designing verification algorithms that leverage duality. To that end, we generalize the concept of a Lagrangian that is commonly used in linear programming and optimization to capture the domains considered in verification problems, which are usually discrete, e.g., powersets of states, predicates, ranking functions, etc. A Lagrangian then induces a primal problem and a dual problem. We devise an abstract primal-dual procedure that simultaneously searches for a primal solution and a dual solution, where the two searches guide each other. We provide sufficient conditions that ensure that the procedure makes progress
under certain monotonicity assumptions on the Lagrangian.

We show that many existing algorithms in program analysis, verification, and automated reasoning can be derived from our algorithmic framework with a suitable choice of Lagrangian. The Lagrangian-based formulation sheds new light on various characteristics of these algorithms, such as the ingredients they use to ensure monotonicity and guarantee progress. We further use our framework to develop a new validity checking algorithm for fixpoint logic over quantified linear arithmetic. 
Our prototype achieves promising results and in some cases solves instances that are not solved by state-of-the-art techniques.

\end{abstract}



\maketitle

\section{Introduction}\label{sec:intro}

Duality and the primal-dual approach are a corner stone in algorithm design and optimization~\cite{Linear-programming1,convex-programming1}.
In program analysis, verification, and automated reasoning, many algorithms also have a primal-dual flavor.
The common duality in these contexts is between a proof and a refutation or counterexample.
That is, the algorithm simultaneously searches for a proof and a refutation, and both searches guide each other in a meaningful way.
For example, in counterexample guided abstraction refinement (CEGAR)~\cite{DBLP:conf/cav/ClarkeGJLV00,DBLP:conf/pldi/BallMMR01},
the search alternates between searching for a counterexample to the current abstraction and refining the abstraction to eliminate a spurious counterexample.
The intuition is that refining the abstraction according to a spurious counterexample makes progress towards a proof (if one exists),
and finding counterexamples for more refined abstractions might make progress towards a non-spurious counterexample (if one exists).

A similar interaction between the search for proofs and the search for counterexamples exists in many program analysis, verification, and automated reasoning algorithms,
including may-must analysis~\cite{Godefroid2010}, ICE learning~\cite{ice}, lazy SMT solving~\cite{dpllt}, quantifier instantiation~\cite{DBLP:conf/cav/GeM09}, program synthesis~\cite{DBLP:conf/icse/JhaGST10}, game solving~\cite{DBLP:journals/pacmpl/FarzanK18}  and more.
In most of these cases, the primal-dual nature of the search algorithm is not made explicit and formal.

Recently, some algorithms were proposed that are explicitly manifested as primal-dual algorithms.
In~\cite{Padon2022}, a primal-dual algorithm for safety verification is proposed, based on a duality between executions in a transition system and a form of incremental induction proofs.
In~\cite{Unno2023}, a primal-dual algorithm for validity of formulas in first-order fixpoint logic with background theories is proposed,
based on a duality between the validity of a formula and that of its negation, which corresponds to a duality between least and greatest fixpoints.
However, while these algorithms are both explicitly primal-dual, the duality and the primal-dual algorithm are bespoke in each case.

In this paper, we propose a unifying perspective that shows that many algorithms in program analysis, verification, and automated reasoning
that leverage duality, be it informally and intuitively or explicitly and formally, can be derived from a common principle.
Moreover, this principle is connected to the classical notion of duality based on a Lagrangian used in linear programming~\cite{Linear-programming1,convex-programming1}.
We leverage the Lagrangian perspective to study different characteristics of existing algorithms,
and to develop a new algorithm for solving formulas in fixpoint logic over quantified linear arithmetic.

In linear programming, a Lagrangian $L(x, \lambda)$ is typically defined as a function from $\mathbb{R}^n \times \mathbb{R}^m$ to $\mathbb{R}$.
The Lagrangian induces two related optimization problems:
the \emph{primal} problem is finding $\inf_x \sup_\lambda L(x, \lambda)$,
and the \emph{dual} problem is finding $\sup_\lambda \inf_x L(x, \lambda)$.
It is easy to see that $\sup_\lambda \inf_x L(x, \lambda) \leq \inf_x \sup_\lambda L(x, \lambda)$,
which is known as \emph{weak duality} (\emph{strong duality} is when an equality holds).
That is, the solution to the dual problem provides a lower bound on the solution of the primal problem.
Moreover, for any $\lambda$, the value of $\inf_x L(x, \lambda)$ provides such a lower bound.
Similarly (dually), the value of $\sup_\lambda L(x, \lambda)$ for any $x$ provides an upper bound on the dual problem.
These properties are used in algorithm design to simultaneously search for solutions to the primal and dual problems, where both searches guide each other.

We propose to define a Lagrangian $L(x,y)$ from $X \times Y$ for arbitrary sets $X$ and $Y$ to a totally-ordered complete lattice,
such that the induced primal and dual problems correspond, intuitively, to checking validity of a candidate proof and checking validity of a candidate refutation.
Unlike in the linear programming case, the sets $X$ and $Y$ that arise in verification contexts are not real vector spaces but
discrete sets or lattices (e.g, powerset lattices of sets of states or abstract domains).

For example, we can define a Lagrangian that captures the classic CEGAR algorithm as follows.
The set $X$ represents possible (abstract) counterexample traces, i.e, $X$ is the set of sequences of states.
The set $Y$ represents possible abstractions, i.e., $Y$ is the set of finite sets of predicates.
We define the Lagrangian $L((s_0 s_1 \cdots s_n), \{p_0, p_1, \ldots, p_m\})$ to be $-1$ if the sequence $(s_0 s_1 \cdots s_n)$
is a counterexample for the abstract transition system defined by the abstraction to $\{p_0, p_1, \ldots, p_m\}$, and $1$ otherwise.
That is, the Lagrangian is $-1$ if for every $0<i<n$ there is a transition from a state $s$ to a state $t$ such that
on all the predicates in $\{p_0, p_1, \ldots, p_m\}$, $s$ agrees with $s_i$ and $t$ agrees with $s_{i+1}$.
With this definition of a Lagrangian, safety is captured by $\sup_y \inf_x L(x, \lambda) = 1$,
which holds if there exists a set of predicates $y$ such that for every sequence of states $x$,
$L(x,y) = 1$, i.e., $x$ is not an abstract counterexample in the abstract transition system induced by $y$.

We further propose an algorithmic framework that given a Lagrangian, tries to solve the primal and dual problems simultaneously with primal-dual guidance between both sides of the duality.
The basic algorithm attempts to find $\alpha \in X$ such that $\sup_y L(\alpha, y)$  is small enough,
providing an upper bound on the primal solution, %
or $\beta \in Y$ such that $\inf_x L(x,\beta)$ is large enough, providing a lower bound on the dual problem.
For example, in the case of CEGAR, if $\inf_x L(x,\beta)$ is $1$ for some set of predicates $\beta$, so is the solution of the dual problem, and the system is safe.
Starting from some candidates $\alpha$ and $\beta$, as long as the aforementioned conditions do not hold, the algorithm iteratively
updates $\alpha$ to some $\alpha'$ that witnesses that $\inf_x L(x,\beta)$ is not large enough (i.e., $L(\alpha',\beta)$ is not large enough);
or updates $\beta$ to $\beta'$ that witnesses that $\sup_y L(\alpha, y)$ is not small enough (i.e., $L(\alpha,\beta')$ is not small enough). That is, the update of the primal candidate $\alpha$ is guided by the dual problem and vice versa.
For CEGAR, given a set of predicates $\beta$, computing a witness for $\inf_x L(x, \beta) \not \geq 1$, i.e., finding $\alpha'$ such that $L(\alpha',\beta)=-1$, corresponds to finding an abstract counterexample for the abstraction $\beta$ defines;
and for a given sequence $\alpha$, finding $\beta'$ such that $L(\alpha,\beta')=1$ corresponds to finding an abstraction that rules out $\alpha$ as an abstract counterexample, i.e., an abstraction refinement step.

Often, $X$ or $Y$ (or both) are join semilattices,
in which case it is possible to increase $\alpha$ or $\beta$ (or both) in each iteration by joining their new value with the previous one.
For example, in the case of the Lagrangian that captures CEGAR, the set $Y$ of possible abstractions (sets of predicates) is a join semilattice,
and increasing $\beta \in Y$ corresponds to ``accumulating'' predicates, as opposed to starting from scratch with a new set of predicates in each iteration.
We prove that if the Lagrangian is (anti-)monotone in $X$ or $Y$ %
the ``accumulation'' of values ensures that the primal-dual algorithm makes progress
in the sense of never generating the same $\alpha$ or $\beta$ more than once.
This is the case for CEGAR, where the Lagrangian is monotone on $Y$.

We show that many existing algorithms for a variety of problems can be described as instances of our framework with a suitable Lagrangian.
For safety verification, we provide Lagrangians that capture CEGAR~\cite{DBLP:conf/cav/ClarkeGJLV00,DBLP:conf/pldi/BallMMR01} (as described above), ICE learning~\cite{ice}, and primal-dual Houdini~\cite{Padon2022}.
We show that our Lagrangian-based algorithmic framework can also capture algorithms
for termination verification~\cite{Podelski2004}, and for solving quantified linear real arithmetic (LRA) formulas~\cite{Farzan2016}.

The Lagrangians succinctly capture the corresponding algorithms, and make the primal-dual interaction explicit.
Furthermore, they allow us to compare the different algorithms.
For example, our Lagrangian-based framework makes it clear that CEGAR makes progress by accumulating on the proof side,
ICE makes progress by accumulating on the counterexample side,
and primal-dual Houdini accumulates on both sides.

In some cases, the Lagrangians shed new light on some of the ingredients of existing techniques.
For example, for termination verification, the Lagrangian perspective exposes that the disjunction found in disjunctively well-founded transition invariants~\cite{Podelski2004} provides a mechanism to accumulate and make monotonic progress on the proofs side for termination proofs.
For solving quantified LRA formulas, we see that the strategy skeletons of~\cite{Farzan2016} play a similar role in enabling monotonic progress.

Another interesting aspect exposed by our unifying framework is the symmetry or asymmetry between the primal and dual problems induced by the Lagrangian.
For example, in primal-dual Houdini and in quantified LRA solving, the primal and dual problems are symmetric in the sense that they have an identical high-level structure. In contrast, in CEGAR and ICE the two sides of the duality have different structures.

The Lagrangians also provide a new perspective on the cases where the algorithms captured by them fail,
due to a connection between strong duality and the ability of the generic primal-dual algorithm to terminate.
Namely, in cases where the primal and dual solutions do not coincide (\ie, strong duality does not hold), the primal-dual algorithm diverges.
In the kinds of problems we consider, strong duality often holds for an ``idealized'' version of the problem, where sets of states can be infinite,
or where any set of states can be captured by a predicate, but is lost when restrictions are imposed on the values in $X$ and $Y$.
For example, in the case of CEGAR, this may be the case if the target system is safe, but the set of predicates considered is not sufficient to prove its safety.
Thus, from the Lagrangian perspective, the loss of precision or expressiveness that prevents termination is captured by the loss of strong duality.

Equipped with Lagrangian based framework and the insights it provides, we go on to develop a new primal-dual algorithm for solving formulas in fixpoint logic over quantified linear arithmetic. Our new algorithm combines ideas from the aforementioned algorithms for termination verification and quantified LRA solving. We derive it cleanly by defining a Lagrangian that combines elements from the Lagrangians of these prior algorithms. A prototype implementation of this new algorithm achieves promising results and solves some instances that state-of-the-art algorithms do not. 

In summary, this paper makes the following contributions:
\begin{enumerate}
\item We define the notions of a Lagrangian and the primal and dual optimization problems it induces in a way that generalizes the classic Lagrangian from linear programming and is general enough for the setting of verification algorithms.
\item We present a primal-dual search procedure that is parameterized by a Lagrangian and attempts to conclude either an upper bound on the value of the primal problem or a lower bound on the value of the dual problem.
  We provide a sufficient condition that ensures that the procedure makes progress.
\item We demonstrate that with suitable Lagrangians, several existing algorithms for a variety of program analysis and automated reasoning problems can be seen as instances of our framework,
and some of their characteristics can be explained in terms of the Lagrangians.
  \item We use our framework to derive a new primal-dual algorithm for solving formulas in fixpoint logic over quantified linear arithmetic. We present an initial empirical evaluation that shows the promise of the new algorithm.
\end{enumerate}

The rest of the paper is organized as follows. \Cref{sec:lagrangian} recalls the classical notion of a Lagrangian from linear programming, and introduces our generalization that is applicable to program verification as well as an abstract primal-dual search algorithm that is parameterized by a Lagrangian.
\Cref{sec:safety} considers the safety verification problem and derives CEGAR and ICE learning as instances of our framework.
\Cref{sec:houdini} derives primal-dual Houdini~\cite{Padon2022} as an instance of our framework with a suitable Lagrangian.
\Cref{sec:termination} formalizes algorithms for termination verification based on disjunctive well-founded ranking functions~\cite{Podelski2004} as instances of the primal-dual procedure with a suitable Lagrangian, and \Cref{sec:qlra} introduces a Lagrangian that captures the algorithm of~\cite{Farzan2016} for solving quantified linear real arithmetic (LRA) formulas.
\Cref{sec:fixed-point-logic} presents our new algorithm for fixpoint logic over quantified linear arithmetic.
We conclude the paper with a discussion of related work in \Cref{sec:related}.

\section{Lagrange Duality for Linear Programming and Verification}\label{sec:lagrangian}
Duality is a fundamental and useful concept in linear programming problems, and this section aims to generalize this concept to verification problems.
However, there are various differences between linear programming problems and verification problems.
For example, in a linear programming problem, both the values to be controlled and to be optimized are continuous, whereas in a verification problem, both are often discrete.
This section introduces the Lagrange duality in a general form applicable to both linear programming and verification problems.
We also provide a basic procedure to solve the optimization problem with a progress property.

\subsection{Duality in Linear Programming}
We briefly review the duality in linear programming.
Let \( \Real \) be the set of reals and \( \PositiveReal \) be the subset of non-negative reals.
For vectors \( x,y \in \Real^n \), we write \( \InnerProd{x}{y} \) for the inner product and \( x \le y \) to mean \( x_i \le y_i \) for every \( i \).
For a matrix \( A \in \Real^{n \times m} \), its transpose is written as \( A^T \).

Consider the following linear program, which we shall call the \emph{primal} problem:\footnote{%
  We implicitly assume that \( \{ x \in \PositiveReal^n \mid A x = b \} \) is non-empty and the optimal value is finite.
  }
\begin{align*}
  & \mbox{minimize}\quad \langle c, x \rangle, \qquad x \in \PositiveReal^n, \\
  & \mbox{subject to}\quad A x = b,
\end{align*}
where \( c \in \mathbb{R}^n \), \( b \in \mathbb{R}^m \) and \( A \in \mathbb{R}^{n \times m} \) are constants.
We say \( x \in \PositiveReal^n \) is \emph{feasible} if \( A x = b \).
A feasible \( x \) gives an upper bound \( \langle c, x \rangle \) of the optimal value, but a way to obtain a lower bound is not obvious.

The \emph{dual} problem is a useful tool to overcome the situation.
It is defined as:
\begin{align*}
  & \mbox{maximize}\quad \langle \lambda, b \rangle, \qquad \lambda \in \mathbb{R}^m, \\
  & \mbox{subject to}\quad A^T \lambda \le c.
\end{align*}
An important point is that the optimal value of this dual problem coincides with the optimal value of the primal problem described above.
So \( \lambda \) satisfying \( A^T \lambda \le c \) provides a lower bound \( \InnerProd{\lambda}{b} \) of the optimal value of both the primal and dual problems.
It is often the case that solving the primal problem takes a long time but the dual problem can be solved quickly (and vice versa).
Linear programming algorithms often leverage this duality by exchanging information between the primal and the dual problems.

The primal and dual problems are connected by the \emph{Lagrangian},
a function on \( x \) and \( \lambda \) given by
\begin{equation*}
  L(x, \lambda) \quad:=\quad \langle c, x \rangle - \langle \lambda, A x - b \rangle, \qquad x \in \mathbb{R}_{\ge 0}^n,\; \lambda \in \mathbb{R}^m.
\end{equation*}
This is obtained by changing the \emph{hard constraint} \( A x = b \) to a \emph{soft constraint}.
The constraint is violated if \( Ax - b \) is non-zero for some element, say \( (Ax - b)_i \neq 0 \), and then the objective function is increased by \( -\lambda_i (Ax - b)_i \).
Since this penalty can be made arbitrarily large by choosing \( \lambda_i \) to have the same sign as \( (Ax - b)_i \) and a large absolute value, the worst choice of the weights \( \lambda \) for the penalty makes the objective function \( \infty \) if the constraint \( A x = b \) is violated.
So
\begin{equation*}
  \sup_{\lambda} L(x, \lambda) \quad=\quad
  \begin{cases}
    \langle c, x \rangle &\mbox{if \( A x = b \)} \\
    \infty &\mbox{if \( A x \neq b \).}
  \end{cases}
\end{equation*}
Hence, the optimal value to the original problem is equivalent to \( \inf_x \sup_\lambda L(x, \lambda) \), and the optimal solution \( x^* \) achieves \( \sup_\lambda L(x^*, \lambda) = \inf_x \sup_\lambda L(x, \lambda) \).

The Lagrangian also characterizes the dual problem.
We have
\begin{align*}
  L(x,\lambda)
  &=
  \langle c, x \rangle - \langle \lambda, A x - b \rangle
  \\
  &=
  \langle c, x \rangle - \langle \lambda, Ax \rangle + \langle \lambda, b \rangle
  \\
  &=
  \langle c, x \rangle - \langle A^T \lambda, x \rangle + \langle \lambda, b \rangle
  \\
  &=
  \langle c - A^T \lambda , x \rangle + \langle \lambda, b \rangle.
\end{align*}
For a fixed \( \lambda \), if \( c - A^T \lambda \) has a negative element, say \( (c - A^T \lambda)_i < 0 \), by choosing large \( x_i > 0 \), the value of \( L(x, \lambda) \) can be made arbitrary small.
So
\begin{equation*}
  \inf_x L(x, \lambda) \quad=\quad
  \begin{cases}
    \langle \lambda, b \rangle &\mbox{if \( A^T \lambda \le c \)} \\
    -\infty &\mbox{if \( A^T \lambda \not\le c\).}
  \end{cases}
\end{equation*}
Hence \( \sup_\lambda \inf_x L(x, \lambda) \) is the optimum value for the dual problem.

The coincidence of the primal and dual is expressed as
\begin{equation*}
  \textstyle
  \inf_x \sup_\lambda L(x,\lambda)
  \quad=\quad
  L(x^*, \lambda^*)
  \quad=\quad
  \sup_\lambda \inf_x L(x,\lambda),
\end{equation*}
where \( x^* \) and \( \lambda^* \) are the optimal solutions to the primal and dual problems, respectively.
Furthermore, optimal solutions \( x^* \) and \( \lambda^* \) are characterized by \( \sup_\lambda L(x^*, \lambda) = L(x^*, \lambda^*) = \inf_x L(x, \lambda^*) \).
This non-trivial but useful property is called the \emph{strong duality}.

\subsection{Generalized Lagrange Duality}
In typical situations in linear programming and convex optimization, a Lagrangian takes vectors and returns a real.
Vectors and reals have useful operations and properties, and the development of Lagrange duality exploits these operations and properties.
On the contrary, the verification community mainly deals with logical expressions, state transition systems, \emph{etc.}, which significantly differ from reals and vectors.
This subsection develops the Lagrange duality that works for such situations with less structure.

The following definition is probably the minimum requirement, assuming only that the upper and lower bounds in the codomain of a Langrangian make sense.
\begin{definition}[Lagrangian]
  A \emph{Lagrangian} is a function
  \[
  L \colon X \times Y \longrightarrow P
  \]
  from sets \( X \) and \( Y \) to a totally-ordered complete lattice \( P = (P, \le) \).
  \qed
\end{definition}
Similar to the case of linear programming, a Lagrangian \( L \colon X \times Y \longrightarrow P \) induces primal and dual objective functions as well as the associated optimization problems.
\begin{definition}
  Let \( L \colon X \times Y \longrightarrow P \) be a Lagrangian.
  The \emph{primal objective function} is \( \alpha \mapsto \sup_{y\in Y} L(\alpha, y) \) and the \emph{primal optimization problem} is to minimize the primal objective function, \ie~to compute \( \inf_x \sup_y L(x,y) \).
  Similarly, the \emph{dual objective function} is \( \beta \mapsto \inf_{x \in X} L(x, \beta) \) and the \emph{dual optimization problem} is to compute \( \sup_y \inf_x L(x,y) \).
  \qed
\end{definition}

\begin{example}\label{eg:lagrangian:primal-and-dual-optimization-problems}
  Let \( L_1 \colon \Int \times \Int \longrightarrow \{-1,1\} \) (with \( -1 \le 1 \))
  be the Lagrangian defined by
  \begin{equation*}
    L_1(x,y) :=
    \begin{cases}
      -1 & \mbox{if \( x \ge y \)} \\
      1 & \mbox{if \( x < y \).}
    \end{cases}
  \end{equation*}
  The primal objective function \( \alpha \mapsto \sup_{y \in \Int} L_1(\alpha, y) \) is the constant function to \( 1 \), so the optimal outcome for the primal optimization problem is \( 1 \).
  The dual objective function \( \beta \mapsto \inf_{x \in \Int} L_1(x, \beta) \) is the constant function to \( -1  \), so the optimal outcome for the dual optimization problem is \( -1 \).
  Let \( L_2 \colon \Int \times (\Int \to \Int) \longrightarrow \{-1,1\} \) be another Lagrangian defined by
  \begin{equation*}
    L_2(x,y) :=
    \begin{cases}
      -1 & \mbox{if \( x \ge y(x) \)} \\
      1 & \mbox{if \( x < y(x) \).}
    \end{cases}
  \end{equation*}
  The primal objective function \( \alpha \mapsto \sup_{y \in (\Int \to \Int)} L_2(\alpha, y) \) is the constant function to \( 1 \), so the optimal outcome for the primal optimization problem is \( 1 \).
  The dual objective function \( \beta \mapsto \inf_{x \in \Int} L_2(x, \beta) \) is \( 1 \) if \( \forall x \in \mathbb{Z}. x < \beta(x) \) and \( -1 \) otherwise.
  For example, its value on \( \beta^*(x) = x+1 \) is \( 1 \).
  So the optimal outcome for the dual optimization problem is \( 1 \).
  \qed
\end{example}

As \cref{eg:lagrangian:primal-and-dual-optimization-problems} shows, the optimal values for the primal and dual problems do not necessarily coincide.
However, even this very general setting enjoys a weak form of duality.
\begin{lemma}\label{lem:lagrangian:weak-duality}
  Every Lagrangian \( L \colon X \times Y \longrightarrow P \) enjoys the weak duality:
  \begin{equation*}
    \sup_{y \in Y} \inf_{x \in X} L(x,y)
    \quad\le\quad
    \inf_{x \in X} \sup_{y \in Y} L(x,y).
  \end{equation*}
\end{lemma}
\begin{proof}
  Since \( L(x,y) \le \sup_x L(x,y) \) for every \( x \) and \( y \), we have \( \inf_y L(x,y) \le \inf_y \sup_x L(x,y) \) for every $x$.
  So \( \sup_x \inf_y L(x,y) \le \inf_y \sup_x L(x,y) \).
\end{proof}

\begin{example}\label{eg:lagrangian:weak}
  Recall the Lagrangian \( L_1 \) in \cref{eg:lagrangian:primal-and-dual-optimization-problems}.
  We have
  \begin{equation*}
    \sup_{y \in \Int} \inf_{x \in \Int} L_1(x,y) = -1 < 1 = \inf_{x \in \Int} \sup_{y \in \Int} L(x,y),
  \end{equation*}
  so the inequality in \cref{lem:lagrangian:weak-duality} can be strict.
  In general, for \( L \colon X \times Y \longrightarrow \{ -1,1 \} \), we have \( \inf_{x} \sup_{y} L(x,y) = 1 \) if and only if \( \forall x. \exists y. L(x,y)=1 \) holds, and similarly \( (\sup_{y} \inf_{x} L(x,y) = 1) \Leftrightarrow (\exists y. \forall x. L(x,y) = 1) \).
  So the weak duality is a mild generalization of the well-known fact \( (\exists x. \forall y. \varphi(x,y)) \Rightarrow (\forall y. \exists x. \varphi(x,y)) \).
  \qed
\end{example}

  We say
a Lagrangian \( L\colon X \times Y \longrightarrow P \) enjoys \emph{strong duality} if \( \sup_y \inf_x L(x,y) = \inf_x \sup_y L(x,y) \).
  The Lagrangian \( L_2 \) in \cref{eg:lagrangian:primal-and-dual-optimization-problems} enjoys strong duality.
  The strong duality is a desirable and useful property, but we do not assume it.

  In the examples in this paper, a typical situation is as follows.
  A Lagrangian \( L \colon X \times Y \longrightarrow P \) has an ``idealization'' \( L' \colon X' \times Y' \longrightarrow P \) satisfying \( X \subseteq X' \) and \( Y \subseteq Y' \), and the idealization \( L' \) enjoys the strong duality and its optimal value coincides with the answer to a problem of interest (\eg~whether a given system is safe).
  However, \( X' \) and \( Y' \) are often not computationally tractable sets such as the set \( \powerset(\mathit{States}) \) of all subsets of an infinite set \( \mathit{States} \).
  The Lagrangian \( L \) is a computably tractable approximation of \( L' \), typically obtained by replacing an intractable set (\eg~\( \powerset(\mathit{States}) \)) with a tractable one (\eg~the set \( \finitepowerset(\mathit{States}) \) of finite subsets or an appropriate set of logical formulas describing properties on \( \mathit{States} \)).
  This approximation, however, may lose
  desirable properties such as strong duality, so we do not assume the strong duality of \( L \).

\subsection{A Primal-Dual Procedure}
In the typical setting in this paper, we want to know whether the optimal value of the primal (or dual) optimization problem of a given Lagrangian \( L \) is bounded by a specific value, which we write as \( 0 \).
For example, for the Lagrangian \( \Lcegar \colon X \times Y \longrightarrow \{ -1,1 \} \) in \cref{sec:cegar}, \( \inf_x \Lcegar(x,\beta) = 1 \) if and only if \( \beta \) witnesses the safety of the target system, so we are interested in whether the optimal value \( \sup_y \inf_x \Lcegar(x,y) \) of the dual optimization problem exceeds \( 0 \) or not, regarding the codomain of \( \Lcegar \) as \( \{ -1 < 0 < 1 \} \).
This subsection develops a general procedure to solve this problem.

Suppose that we would like to know whether \( \sup_y \inf_x L(x,y) \ge 0 \) for a given Lagrangian \( L \colon X \times Y \longrightarrow \{-1,1\} \).
To affirmatively answer this question, it suffices to find a witness \( \beta \in Y \) such that \( \inf_x L(x,\beta) \ge 0 \).
Similarly, a negative answer is confirmed by finding \( \alpha \in X \) such that \( \sup_y L(\alpha, y) \leq 0 \).
To find a witness \( \alpha \) or \( \beta \), we start by choosing an arbitrary value \( \beta_0 \in Y \) as a candidate for a witness \( \beta \).
If \( \inf_x L(x, \beta_0) \ge 0 \),
we are done; otherwise, there exists \( \alpha_1 \in X \) such that \( L(\alpha_1, \beta_0) < 0 \).
Now \( \alpha_1 \) can be a candidate for a negative witness \( \alpha \), so we check whether \( \sup_y L(\alpha_1, y) \le 0 \).
We are done if it is the case, and otherwise, \( L(\alpha_1, \beta_1) > 0 \) for some \( \beta_1 \in Y \), which is the next candidate for a positive witness.
In this way, we can iteratively update the candidates of \( \alpha \) and \( \beta \) against each other until we finally find a positive or negative witness.

While the above process might converge, it might also get stuck in a cycle.
For example, if $\beta_1$ above is not a valid positive witness we will discover $\alpha_2$ such that \( L(\alpha_2, \beta_1) < 0 \)  as a counter-witness to $\beta_1$ and our next candidate negative witness.
But it may be the case that \( L(\alpha_2, \beta_0) \ge 0 \), in which case we may start repeating ourselves.
That is, our sequence of candidate positive and negative witnesses will be $\beta_0, \alpha_1, \beta_1, \alpha_2, \beta_0, \ldots$, iterating back and forth between two pairs of candidate positive and negative witnesses.
One way to ensure progress, at least in the sense of not revisiting the same candidates, is to assume additional structure on $X$ or $Y$ (or possibly both).
Suppose that  \( X \) or \( Y \) is a join semilattice \((X, \sqsubseteq_X, \sqcup_X) \) or \((Y, \sqsubseteq_Y, \sqcup_Y) \) and
that the Lagrangian \(L\) is anti-monotone on \( X \) or monotone on \( Y \).
Then, when updating the candidate positive or negative witness, rather than forgetting the previous candidate we can take the join of the previous candidate and the new counter-witness.
The option of this monotonic update appears in \Cref{alg:lagrangian:procedure}, lines~\ref{line:updatealpha} and \ref{line:updatebeta}.
As we prove below, a monotonic update on either side is sufficient to prevent the search from getting stuck and ensures it keeps exploring new possible solutions (on both sides).

\begin{wrapfigure}{L}{0.38\textwidth}
\begin{minipage}{0.38\textwidth}
\begin{algorithm}[H]
  \caption{Primal-Dual Procedure
  }\label{alg:lagrangian:procedure}
  \begin{algorithmic}[1]
    \Require ${\textsc{PrimalDual}}(L)$
      \State let $\alpha \in X$, $ \beta \in Y $
      \While{$\textbf{true}$}
        \If {$ \inf_x L(x,\beta) \ge 0 $} \label{line:primalcheck}
          \State $ \mathbf{return}\;(\mathtt{D}, \beta) $
        \EndIf
        \State let $ \delta \in \{ x \in X \mid L(x, \beta) < 0 \} $ \label{line:delta}
        \State $ \alpha \leftarrow \delta $ or $ \alpha \leftarrow \alpha \sqcup_X \delta $  \label{line:updatealpha}
        \If {$ \sup_y L(\alpha,y) \le 0 $} \label{line:dualcheck}
          \State $ \mathbf{return}\;(\mathtt{P}, \alpha) $
        \EndIf
        \State let $ \gamma \in \{ y \in Y \mid L(\alpha, y) > 0 \} $ \label{line:gamma}
        \State $ \beta \leftarrow \gamma $ or $ \beta \leftarrow \beta \sqcup_Y \gamma $ \label{line:updatebeta}
      \EndWhile
    \end{algorithmic}
\end{algorithm}
\end{minipage}
\end{wrapfigure}

\Cref{alg:lagrangian:procedure} formalizes this idea.
\Cref{alg:lagrangian:procedure} can be divided into two parts, namely lines \ref{line:primalcheck}--\ref{line:delta} and lines \ref{line:dualcheck}--\ref{line:gamma}.
The former checks whether the current \( \beta \) is a positive witness and, if it is not the case, produces a counter \( \delta \).
The latter is
the dual of the former, checking whether \( \alpha \) is a negative witness and producing a counter \( \gamma \).
We call the former the \emph{dual witness check} and the latter the \emph{primal witness check}
(recall that the primal optimization problem is minimization, and the dual is maximization).
\Cref{alg:lagrangian:procedure} does not tell us how to implement the subprocedures to solve these subproblems and leaves the choice in lines~\ref{line:delta} and~\ref{line:gamma} nondeterministic.

\begin{remark}
  There are similar variants of \cref{alg:lagrangian:procedure}.
  E.g.,
  rather than checking the dual witness first, it is possible to start with the primal witness check.
  Also, instead of computing \( \gamma \in \{ y \in Y \mid L(\alpha,y) \not\le 0 \} \) and updating \( \beta \leftarrow \beta \sqcup \gamma \), we can directly update \( \beta \) to an element from \( \{ y \in Y \mid L(\alpha,y) \not\le 0, y \ge \beta \} \).
  We will refer to \Cref{alg:lagrangian:procedure} and these variants collectively as the \emph{basic primal-dual procedure}.
  \qed
\end{remark}

\begin{remark}
  Developing a practical procedure based on \cref{alg:lagrangian:procedure} often requires optimizations.
  Heuristics for choosing a ``good'' \( \delta \) or \( \gamma \) on lines \ref{line:delta} or \ref{line:gamma} can have a significant impact on performance.
  One could also consider managing additional information beyond \( \alpha \) and \( \beta \) to solve the primal and dual witness check problems faster.
  \Cref{alg:lagrangian:procedure} just describes a skeleton of practical procedures.
  \qed
\end{remark}

\subsubsection*{Partial Correctness and Progress}
The partial correctness (i.e., correctness assuming termination) and progress properties of \Cref{alg:lagrangian:procedure} are formalized in the following theorems.

\begin{theorem}\label{thm:lagrangian:primal-dual-method-sound}
  \textsc{PrimalDual} in \cref{alg:lagrangian:procedure} is correct in the following sense:
  \begin{itemize}
    \item If \( (\mathtt{P}, \alpha) = \textsc{PrimalDual}(L) \), then \( 0 \) is an upper bound of the optimal value of the primal optimization problem and \( \alpha \) is a witness: \( \inf_x \sup_y L(x,y) \le \sup_y L(\alpha, y) \le 0 \).
    \item If \( (\mathtt{D}, \beta) = \textsc{PrimalDual}(L) \), then \( 0 \) is a lower bound of the optimal value of the dual optimization problem and \( \beta \) is a witness: \( \sup_y \inf_x L(x,y) \ge \inf_x L(x,\beta) \ge 0 \).
  \end{itemize}
\end{theorem}
\begin{proof}
  If the procedure returns \( (\mathtt{P}, \alpha) \), then the value \( \alpha \) has passed the condition in \cref{line:dualcheck}, hence \( \sup_y L(\alpha, y) \le0\).
  The case of \( (\mathtt{D}, \beta) \) is similar.
\end{proof}
\begin{theorem}\label{thm:lagrangian:primal-dual-method-progress}
  Let \( L \colon X \times Y \longrightarrow P \) be a Lagrangian.
  Assume that \( X \) or \( Y \) is a join semilattice, and \( L \) is anti-monotone on \( X \) or monotone on \( Y \).
  Then, \textsc{PrimalDual} in \cref{alg:lagrangian:procedure} enjoys progress, i.e.~it does not assign the same value to \( \alpha \) or \( \beta \) twice, provided that we always choose $\alpha \leftarrow \alpha \sqcup_X \delta$ in \cref{line:updatealpha} (or that we always choose $\beta \leftarrow \beta \sqcup_Y \gamma$ in \cref{line:updatebeta}).
\end{theorem}
\begin{proof}
  We focus on the case where \( Y \) is a join semilattice and \( \beta \) is updated monotonically (the other case is similar).
  We write \( \alpha_n \), \( \beta_n \) and \( \gamma_n \) for the values at the \( n \)-th iteration of the loop.
  Assume for contradiction that \( \alpha_i = \alpha_j \) for some \( i < j \).
  Then \( L(\alpha_i, \gamma_i) \not\le0\).
  Since \( \gamma_i \le \beta_i \le \beta_{j-1} \), we have \( L(\alpha_i, \beta_{j-1}) \not\le0\) by the monotonicity of \( L \) on \( Y \).
  We have \( L(\alpha_j, \beta_{j-1}) \not\ge 0 \) by the definition of \( \alpha_j \), a contradiction.
  We show that \( \beta_i \neq \beta_j \) for every \( i \neq j \).
  Since \( \beta_i \le \beta_{i+1} \le \dots \le \beta_j \) for \( i \le j \), it suffices to show that \( \beta_i < \beta_{i+1} \), or equivalently that \( \beta_i \neq \beta_{i+1} \).
  Assume for contradiction that \( \beta_i = \beta_{i+1} \).
  By the definition of \( \beta_{i+1} \), we have \( L(\alpha_{i+1}, \beta_{i+1}) \not\le0\).
  So \( L(\alpha_{i+1}, \beta_i) \not\le0\) since \( \beta_i = \beta_{i+1} \).
  By the definition of \( \alpha_{i+1} \), we have \( L(\alpha_{i+1}, \beta_{i}) \not\ge 0 \), a contradiction.
\end{proof}

\subsubsection*{Termination}
Unlike partial correctness and progress, termination of \cref{alg:lagrangian:procedure} is generally not guaranteed (except for trivial cases where \( X \) or \( Y \) are finite). This is because \cref{alg:lagrangian:procedure} is a general procedure applicable to a variety of verification problems, many of which are undecidable. Our Lagrangian framework, however, provides a necessary condition for termination, relating termination and strong duality.
\begin{theorem}\label{thm:lagrangian:gap-intermination}
  Assume \( L \colon X \times Y \longrightarrow \{ -1,1 \} \) does not enjoy the strong duality property, \ie{},
  $\inf_x \sup_y L(x,y) = 1$ while $\sup_y \inf_x L(x,y) = -1$.
  Then \cref{alg:lagrangian:procedure} cannot terminate.
  A similar claim holds if \( 0 \) is in the duality gap in the sense that
  \begin{equation*}
    \sup_y \inf_x L(x,y) < 0 < \inf_x \sup_y L(x,y).
  \end{equation*}
\end{theorem}
\begin{proof}
  Since \( \sup_y L(\alpha,y) > 0 \) and \( \inf_x L(x, \beta) < 0 \) for every \( \alpha \in X \) and \( \beta \in Y \), the conditions in \cref{line:primalcheck,line:dualcheck} are never met.
\end{proof}
A typical example of the above situation is CEGAR satisfying the following conditions: the target system is safe, and the set of available predicates is strong enough to refute individual suspicious error traces but not enough to prove the safety of the entire system.

\begin{example}
  \Cref{alg:lagrangian:procedure} may diverge, even for a Lagrangian with strong duality.
  Let \( \Int_{\infty} := \Int \cup \{ \infty \} \) and \( L_3 \colon \Int \times \Int_{\infty} \longrightarrow \{ -1,1 \} \) be the Lagrangian defined by \( L_3(x,y)=1 \Leftrightarrow (x < y) \).
  Note that \( \inf_x L_3(x,\infty) = 1 \) but \( \inf_x L_3(x,\beta) = -1 \) for every \( \beta \neq \infty \).
  So the loop continues until \( \beta \) becomes \( \infty \).
  However, one can always choose a finite \( \gamma \) in \cref{line:gamma}, resulting in divergence of the procedure.
  \qed
\end{example}

Let us now turn our attention to the opposite, \ie~ideas to ensure termination in some situations.
Assume a Lagrangian \( L \colon X \times Y \longrightarrow \{ -1, 1 \} \) with a positive witness \( \beta \) satisfying \( \inf_x L(x,\beta) = 1 \).
\emph{Stratification}~\cite{Jhala2006,Unno2021}
is a technique to ensure the procedure eventually finds a positive witness (possibly different from \( \beta \)).
The idea is to decompose \( Y = \biguplus_{n \in \Nat} Y_n \) into an infinite union of finite sets,
and to force the procedure to find \( \gamma \in Y_l \) with minimum \( l \) (in \cref{line:gamma}). %

By combining monotonicity as formulated in \Cref{thm:lagrangian:primal-dual-method-progress} with stratification we can ensure termination,
assuming a suitable positive or negative witness exists.
The conditions for termination vary between the case where monotonicity and stratification are applied at the same side (i.e., both to \(X \) or to \(Y\)) or at different sides (i.e., \(X\) is a semilattice with monotonicity and \(Y\) is stratified, or vice versa).
Below we analyze the case where \(X\) is a semilattice and \(Y\) is stratified and
then the case where \(Y\) is a semilattice and also stratified.
The other two cases are dual.

\begin{theorem}\label{thm:lagrangian:stratification}
  Let \( L \colon X \times Y \longrightarrow \{ -1, 1 \} \) be a Lagrangian where
  \( X \) is a join semilattice with \(L\) anti-monotone on \( X \),
  and \( Y \) is stratified, i.e., \( Y = \biguplus_{n \in \Nat} Y_n \) such that \( Y_n \) is finite for each \( n \).
  Assume that \( L \) has a positive witness \( \beta^* \), \ie~\( \inf_x L(x,\beta^*) = 1 \).
  Then \cref{alg:lagrangian:procedure} terminates provided that \( \gamma \) in \cref{line:gamma} is chosen from the smallest possible layer, \ie~\( \gamma \in Y_l \) with \( l = \min \{ k \in \Nat \mid \exists \gamma \in Y_k. L(\alpha, \gamma) > 0 \} \).
\end{theorem}
\begin{proof}
  Assume \( \beta^* \in Y_n \).
  Then \( Y_{\le n} := \bigcup_{i = 0}^n Y_i \) is a finite set.
  \Cref{alg:lagrangian:procedure} terminates before \( | Y_{\le n} | \) iterations because of progress (\cref{thm:lagrangian:primal-dual-method-progress}) and because \( \beta \in Y_{\le n} \) is an invariant of the procedure,
  since \( \beta^* \) is always a possible choice in \cref{line:gamma} and we choose \( \gamma \) from the smallest possible layer.
\end{proof}

\Cref{thm:lagrangian:stratification} guarantees termination (when a positive witness exists) but imposes an additional condition on the choice of \( \gamma \), making the implementation of \cref{line:gamma} more difficult.
It is therefore important to design a stratification for which the modified \cref{line:gamma} can be implemented efficiently.

Stratification and monotonicity can also be on the same side (i.e., \(X\) or \(Y\)).
In a stratification of a semilattice the layers need not be finite, only of finite height, and they must be closed under join.

\newcommand{\layer}{r}
\begin{theorem}\label{thm:lagrangian:stratification-lattice}
  Let \( L \colon X \times Y \longrightarrow \{ -1, 1 \} \) be a Lagrangian and assume that \( Y \) is a join semilattice and \( L \) is monotone on \( Y \).
  Assume a join semilattice homomorphism \( \layer\colon Y \longrightarrow \Nat \), \ie~\( \layer(y \sqcup y') = \max(\layer(y), \layer(y')) \) and consider the induced stratification \( Y_n := \{ y \in Y \mid \layer(y) = n \} \).
  Suppose that each \( Y_n \) is of finite height, i.e., it has no infinite increasing chain \( y_0 \sqsubset y_1 \sqsubset \cdots \in Y_n \).
  We assume that the procedure always chooses \( \beta \leftarrow \beta \sqcup_Y \gamma \) in \cref{line:updatebeta} and that \( \gamma \) in \cref{line:gamma} is chosen from the smallest possible layer.
  Then the procedure terminates provided that there exists a positive witness \(\beta^*\).
\end{theorem}
\begin{proof}
  Similar to that of \cref{thm:lagrangian:stratification}, combining progress and the invariant \( \layer(\beta) \leq \layer(\beta^*)\).
\end{proof}

\begin{remark}
  Although we were only interested in the sign of the optimal outcome of a Lagrangian, seeking the optimal outcome can also be useful for verification in some contexts.
  In this setting, it is natural to choose \( \gamma \) in \cref{line:gamma} in \cref{alg:lagrangian:procedure} from those achieveing the best outcome against \( \alpha \) (\ie~\( L(\alpha,\gamma) = \max \{ L(\alpha,y) \mid y \in Y \} \)).
  A Lagrangian \( L \colon X \times Y \longrightarrow \{ -1, 1 \} \) with a stratification \( Y = \biguplus_{n \in \Nat} Y_n \) induces the \emph{stratified Lagrangian} \( \widehat{L} \colon X \times Y \longrightarrow \Real \) defined by
  \begin{equation*}
    \widehat{L}(x,y) \quad:=\quad
    \begin{cases}
      -1 & \mbox{if \( L(x,y) = -1 \)} \\
      1/n & \mbox{if \( L(x,y) = 1 \) and \( y \in Y_n \).}
    \end{cases}
  \end{equation*}
  The procedure seeking the optimal outcome of \( \widehat{L} \) is the stratified version of \cref{alg:lagrangian:procedure} mentioned in \cref{thm:lagrangian:stratification,thm:lagrangian:stratification-lattice}.
  \qed
\end{remark}

\subsubsection*{Summary: Requirements for Applying our Framework}

We conclude this section by summarizing the process to apply \cref{alg:lagrangian:procedure}.
\begin{itemize}
  \item Develop a Lagrangian \( L \colon X \times Y \longrightarrow P \) with the following conditions:
  \begin{itemize}
    \item \textbf{Soundness}: At least one of \( \inf_x \sup_y L(x,y) \le 0 \) and \( \sup_y \inf_x L(x,y) \ge 0 \) is related to the problem we would like to solve.
    \item \textbf{Monotonicity}: At least one of \( X \) or \( Y \) is a join-lattice, and \( L \) is (anti-)monotone on that component.
  \end{itemize}
  \item Develop solvers for both the primal and dual witness check problems.  Possibly design heuristics for the nondeterministic choices in lines~\ref{line:delta} and~\ref{line:gamma} that lead to good generalization in practice.
\end{itemize}

\begin{remark}
  Hereafter, we assume that the codomain \(P\) of a Lagrangian is a \emph{finite subset of\/ \(\Int\)} with the order inherited from \( \Int \).
  All the examples in this paper satisfy this assumption.
  \qed
\end{remark}

\section{Lagrangians for Safety Verification: CEGAR and ICE}\label{sec:safety}
This section provides the first examples of Lagrangians in verification.
We provide two Lagrangians for the safety verification problem, which induce CEGAR and ICE.

\subsection{Problem setting}
This section focuses on the \emph{safety verification problem} of a given transition system, which is a problem asking if one can reach a bad state in a given transition system.
We give a formal definition.

\begin{definition}\label{def:transition-system}
  A \emph{transition system} \( \TSys \) is a tuple \( (\stateSet[\TSys], \initState[\TSys], \trans[\TSys], \badState[\TSys]) \) where \( |\TSys| \) is the set of \emph{states}, \( \initState[\TSys] \subseteq \stateSet[\TSys] \) is the \emph{initial states}, \( ({\trans[\TSys]}) \subseteq \stateSet[\TSys] \times \stateSet[\TSys] \) is the \emph{transition relation}, and \( \badState[\TSys] \subseteq \stateSet[\TSys] \) is the set of \emph{bad states}.
  A state \( s \in \stateSet[\TSys] \) is \emph{reachable} if there exists a transition sequence \( \initState[\TSys] \ni s_0 \trans[\TSys] s_1 \trans[\TSys] \dots \trans[\TSys] s_n = s \) for some \( n \ge 0 \).
  The transition system \( \TSys \) is \emph{safe} if no bad state \( s \in \badState[\TSys] \) is reachable.
  Otherwise \( \TSys \) is \emph{unsafe}.
  \qed
\end{definition}

\begin{example}\label{eg:transition-system}
  Let \( \TSys_0 \) be the transition system given by \( \stateSet[\TSys_0] = \Int \),
  \( \initState[\TSys_0] = \{ 0 \} \), \( (x \trans[\TSys_0] y) \Leftrightarrow x+1=y \) and \( \badState[\TSys_0] = \{ -3 \} \).
  The set of reachable states is \( \{ n \in \Int \mid n \ge 0 \} \),
  so \( \TSys_0 \) is safe.
  \qed
\end{example}

The unsafety of a transition system can be witnessed by a transition sequence \( \initState[\TSys] \ni s_0 \trans[\TSys] s_1 \trans[\TSys] \dots \trans[\TSys] s_n \in \badState[\TSys] \) from an initial state \( s_0 \) to a bad state \( s_n \).
A safety proof cannot always be possible in such an obvious way.
A common approach is to consider an appropriate set of \emph{predicates} \( \PredSet \)
and try to find a \emph{safe inductive invariant}.
\begin{definition}\label{def:predicate-set-for-a-transition-system}
  A \emph{predicate set} over a transition system \( \TSys \) is a set \( \PredSet \) together with a \emph{satisfaction relation} \( ({\models}) \subseteq \stateSet[\TSys] \times \PredSet \).
  When \( (s,p) \in ({\models}) \), we write \( s \models p \) and say that \( s \) \emph{satisfies} \( p \).
  A subset \( X \subseteq \stateSet[\TSys] \) is an \emph{inductive invariant} if it satisfies the following conditions: (1)~\( \initState[\TSys] \subseteq X \) (\emph{initiation}); and (2)~\( (s \in X) \wedge (s \trans[\TSys] s') \Longrightarrow (s' \in X) \) (\emph{consecution}).
  If an invariant further satisfies (3)~\( \badState[\TSys] \cap X = \emptyset  \) (\emph{safety}), we call it a \emph{safe inductive invariant}.
  A predicate \( p \in \PredSet \) is a \emph{(safe) inductive invariant} if so is \( \{ s \in \stateSet[\TSys] \mid s \models p \} \).
  \qed
\end{definition}

\begin{example}
  Recall the transition system \( \TSys_0 \) in \cref{eg:transition-system}.
  An inductive invariant is \( p(x) :\Leftrightarrow (x \ge 0) \).
  There may be more than one inductive invariant for a system.
  For example, \( p'(x) :\Leftrightarrow (x > -2) \) is another inductive invariant for \( \TSys_0 \).
  \qed
\end{example}

\subsection{Lagrangian for CEGAR}\label{sec:cegar}
\emph{Counter-Example Guided Abstraction Refinement} (known as \emph{CEGAR}) is a famous technique for the safety verification problem~\cite{DBLP:conf/cav/ClarkeGJLV00,DBLP:conf/pldi/BallMMR01}.
The CEGAR procedure manages a finite subset \( A \subseteq \PredSet \) of the predicate set, and it iterates the following two phases.
\begin{enumerate}
  \item[(a)] Abstract the target transition system \( \TSys \) by using the predicate set \( A \) and check its safety.
    If the abstract system is safe, then \( \TSys \) is safe; otherwise, the abstract transition system has a transition sequence to a bad state.
  \item[(b)] Check the feasibility of the abstract transition sequence to a bad state.
    If it is feasible, then \( \TSys \) is unsafe; otherwise, add predicates to \( A \) that suffice to show the infeasibility of the abstract transition sequence.
\end{enumerate}

The abstract transition system for a predicate set \( A \subseteq \PredSet \) is given as follows.
The states \( s,s' \in \stateSet[\TSys] \) are \emph{indistinguishable by \( A \)}, written \( s \stackrel{A}{\approx} s' \), if \( (s \models p) \Leftrightarrow (s' \models p) \) for every \( p \in A \).
A state of the abstract transition system is an equivalence class \( [s]_{A} := \{\, s' \in \stateSet[\TSys] \mid s \stackrel{A}{\approx} s' \,\} \) of \( \stackrel{A}{\approx} \).
The abstract transition system has a transition from \( [s_1]_A \) to \( [s_2]_A \) if there exists a transition from a state in \( [s_1]_A \) to a state in \( [s_2]_A \) (\ie~\( [s_1]_A \ni s_1' \trans[\TSys] s_2' \in [s_2]_A \) for some \( s_1' \) and \( s_2' \)).
An abstract state \( [s]_A \) is an initial state (resp.~a bad state) if it contains an initial state (resp.~a bad state).
We write \( \TSys^A = (\stateSet[\TSys]^A, \initState[\TSys]^A, \trans[\TSys]^A, \badState[\TSys]^A) \) for the abstract transition system.

CEGAR can be seen as an instance of the Lagrangian-based primal-dual method.
Let
\begin{equation*}
  X \::=\: \stateSet[\TSys]^*
  \quad\mbox{and}\quad
  Y \::=\: \finitepowerset(\PredSet)
\end{equation*}
where \( \stateSet[\TSys]^* \) is the set of all finite sequences over \( \stateSet[\TSys] \) and \( \finitepowerset(\PredSet) \) is the set of all finite subsets of \( \PredSet \).
We regard \( Y = \finitepowerset(\PredSet) \) as a poset ordered by set-inclusion \( \subseteq \).
The Lagrangian is given by
\begin{equation*}
  \Lcegar((s_0s_1\dots s_n), A)
  \quad:=\quad
  \begin{cases}
    \:\: -1 \quad & \mbox{if \( \initState[\TSys]^A \ni [s_0]_A \trans[\TSys]^A [s_1]_A \trans[\TSys]^A \cdots \trans[\TSys]^A [s_n]_A \in \badState[\TSys]^A \)} \\
    \:\: 1 & \mbox{otherwise.}
  \end{cases}
\end{equation*}
An ``idealization'' \( \Lcegar' \) of \( \Lcegar \) is obtained by setting \( \PredSet \) to be \( \powerset(\stateSet[\TSys]) \).
\begin{proposition}\label{prop:safety:cegar:idealized}
  \( \Lcegar' \) enjoys strong duality, and its optimal value is \( 1 \) if and only if \( \TSys \) is safe.
\end{proposition}
\begin{proof}
  If \( \TSys \) is unsafe, a concrete error trace \( \initState[\TSys] \ni s_0 \trans[\TSys] \dots \trans[\TSys] s_n \in \badState[\TSys] \) is an optimal choice of \( X \).
  Otherwise, the set \( R \) of reachable states is a safe inductive invariant and \( \{ R \} \) is optimal for $Y$.
\end{proof}

\begin{corollary}\label{cor:safety:cegar:soundness}
  If \( \sup_y \inf_x \Lcegar(x,y) = 1 \), then \( \TSys \) is safe.
  \qed
\end{corollary}
If \( \inf_x \sup_y \Lcegar(x,y) = -1 \), then we know that \(Y\) is insufficient to prove the safety of \( \TSys \) but not to say that \( \TSys \) is unsafe.
However, given \( \alpha \) such that \( \sup_y \Lcegar(\alpha,y) = -1 \), we can examine it and may conclude that \( \TSys \) is unsafe.

\cref{cor:safety:cegar:soundness} ensures the soundness criterion.
The monotonicity criterion is trivially met on \( Y \).
We examine the primal and dual witness check problems.
\begin{itemize}
  \item The dual witness check asks to find \( \tau \) such that \( \Lcegar(\tau, A) < 0 \).
    This is equivalent to the safety verification of the abstract transition system \( \TSys^A \), and \( \tau \) is an error trace.
  \item The primal witness check asks to find \( A \) such that \( \Lcegar(\tau, A) > 0 \).
    This is the so-called abstraction refinement and has well-known procedures.
\end{itemize}
So \cref{alg:lagrangian:procedure} with \( \Lcegar \) is applicable to solve the safety verification problem.
\Cref{alg:CEGAR} describes the resulting procedure, which is the standard CEGAR procedure.

\begin{algorithm}[t]
  \caption{~~CEGAR in the standard description (left) and as a primal-dual method (right)}\label{alg:CEGAR}
  \begin{minipage}{.47\linewidth}
    \begin{algorithmic}[1]
      \State let $ A \leftarrow \emptyset $
      \While{$\textbf{true}$}
        \If {$ \TSys^A $ is safe}
          \State $ \mathbf{return}\;(\mathtt{safe}, A) $
        \EndIf
        \State let $ \tau $ be an abstract error trace in $\TSys^A$
        \If {feasibility of $ \tau $ is not refutable}
          \State $ \mathbf{return}\;(\mathtt{unknown}, \tau) $
        \EndIf
        \State let $ \alpha $ be an infeasibility witness for $\tau$
        \State $ A \leftarrow A \cup \alpha $
      \EndWhile
    \end{algorithmic}
  \end{minipage}
  \begin{minipage}{.47\linewidth}
    \begin{algorithmic}[1]
      \State let $ A \leftarrow \emptyset $
      \While{$\textbf{true}$}
        \If {$ \inf_{\tau} \Lcegar(\tau, A) \ge 0 $}
          \State $ \mathbf{return}\;(\mathtt{safe}, A) $
        \EndIf
        \State let $ \tau \in \{ \tau \in \stateSet[\TSys]^* \mid \Lcegar(\tau, A) \not\ge 0 \} $
        \If {$ \sup_\alpha \Lcegar(\vec{s}, \alpha) \le 0 $}
          \State $ \mathbf{return}\;(\mathtt{unknown}, \tau) $
        \EndIf
        \State let $ \alpha \in \{ \alpha \in Y \mid  \Lcegar(\tau, \alpha) \not\le 0 \} $
        \State $ A \leftarrow A \cup \alpha $
      \EndWhile
    \end{algorithmic}
  \end{minipage}
\end{algorithm}

\begin{remark}
  The stratified version of CEGAR~\cite{Jhala2006} can be understood in the Lagrangian framework as follows.
  In this setting, the set of predicates is stratified: \( \PredSet = \biguplus_{n \in \Nat} \PredSet_n \).
  This induces a stratification on \( Y =  \finitepowerset(\PredSet) \) as follows: \( A \in \finitepowerset(\PredSet) \) belongs to the \( n \)-th level where \( n \) is the minimum number such that \( A \subseteq \bigcup_{i = 0}^n \PredSet_i \).
  By \cref{thm:lagrangian:stratification}, the CEGAR procedure is guaranteed to terminate for instances with positive witnesses, as long as newly added predicates are always chosen from the smallest possible level.
  \qed 
\end{remark}

\subsection{Lagrangian for ICE}\label{sec:ice}
The \emph{ICE learning}~\cite{ice} is another approach to the safety verification problem.
It is described as an interaction between a \emph{teacher} and a \emph{learner}.
The teacher tells a finite information on the system, and the leaner generates a candidate of satefy proof that works for the provided information.
Then the teacher checks if the candidate is a proof against the whole system, and if not, tells an additional information on the system that explains why the candidate generated by the leaner is insufficient.
\Cref{eg:safety:ice} explains the interaction by an example.
\begin{example}\label{eg:safety:ice}
  We prove the safety of \( \TSys_0 \) in \cref{eg:transition-system} by ICE.
  Information passed from the teacher to the learner is a triple \( (I', \trans', B') \) of finite subsets of \( \initState[\TSys_0] \), \( \trans[\TSys_0] \) and \( \badState[\TSys_0] \), respectively.
  Let \( S'_0 = (I'_0, \trans[0]', B'_0) = (\emptyset, \emptyset, \emptyset) \) be the initial choice.
  The learner proposes an inductive invariant for the subsystem \( S'_0 \): the subsystem \( S'_0 \) does not have a bad state, so \( p_0(x) :\Leftrightarrow \top \) is an inductive invariant for the subsystem \( S'_0 \).
  Then the teacher checks if \( p_0 \) is an inductive invariant for \( \TSys_0 \), but it is obviously not.
  The teacher adds a bad state \( -3 \) to the subsystem, resulting in \( S'_1 = (\emptyset, \emptyset, \{ -3 \}) \).
  The leaner proposes an inductive invariant for the subsystem \( S'_1 \): the subsystem does not have an initial state, so \( p_1(x) :\Leftrightarrow \bot \) is an inductive invariant for the subsystem \( S'_1 \).
  The teacher then checks if \( p_1 \) is an inductive invariant for \( \TSys_0 \), but it is not.
  The teacher adds an initial state \( 0 \) to the subsystem: \( S'_2 = (\{ 0 \}, \emptyset, \{ -3 \}) \).
  The learner proposes an inductive invariant for the subsystem \( S'_2 \): one can choose \( p_2(x) :\Leftrightarrow (x \mathbin{\mathrm{mod}} 2 = 0) \), which separates the initial state \( 0 \) from the bad state \( -3 \).
  The teacher then checks if \( p_2 \) is an inductive invariant for \( \TSys_0 \), but it is not because \( p_2 \) is not closed under the transition relation.
  The teacher adds some examples of transitions, yielding \( S'_3 = (\{ 0 \}, \{ (0,1), (2,3), (8,9), (-4, -3) \}, \{ -3 \}) \).
  Now, the learner has a much narrower range of candidates: an inductive invariant \( p_3 \) for \( S'_3 \) must be true on \( 0,1 \) and must be false on \( -4, -3 \).
  A candidate is \( p_3(x) :\Leftrightarrow (x > -2) \).
  The teacher can ensure that \( p_3 \) is an inductive invariant for %
  \( \TSys_0 \).
  \qed
\end{example}

Then ICE-learning can be seen as an instance of the Lagrangian-based primal-dual method.
Let
\begin{equation*}
  X \::=\: \finitepowerset(\initState[\TSys]) \times \finitepowerset(\trans[\TSys]) \times \finitepowerset(\badState[\TSys])
  \quad\mbox{and}\quad
  Y \::=\: \PredSet.
\end{equation*}
Intuitively, \( X \) tells us partial information on the transition system.
The Lagrangian is given by
\begin{equation*}
  \Lice(S', p) =
  \begin{cases}
    \:\: 1 \quad & \mbox{\( p \) is an inductive invariant for the subsystem \( S' \) of \( \TSys \)} \\
    \:\: -1 & \mbox{otherwise.}
  \end{cases}
\end{equation*}
So \( \Lice(S', \varphi) \) is $1$ if and only if \( \varphi \) witnesses the safety of the subtransition system \( S' \).
The ``idealization'' \( \Lice' \) of \( \Lice \) is obtained by setting \( Y := \powerset(\stateSet[\TSys]) \).
\begin{proposition}
  \( \Lice' \) enjoys the strong duality, and its optimal value is \( 1 \) iff \( \TSys \) is safe.
\end{proposition}
\begin{proof}
  If \( \TSys \) is safe, an optimal choice of \( Y \) is the set of reachable states.
  If \( \TSys \) is unsafe, an optimal choice of \( X \) is an error trace.
\end{proof}
\begin{corollary}
  If \( \sup_y \inf_x \Lice(x,y) = 1 \), then \( \TSys \) is safe.
  \qed
\end{corollary}

The Lagrangian \(\Lice\) satisfies both soundness and monotonicity criteria.
Unlike in the CEGAR case, the \( X \) component has a join semilattice structure and \( \Lice \) is anti-monotone on \( X \).
We examine the primal and dual witness checks.
\begin{itemize}
  \item The dual witness check asks to find \( S' \) such that \( \Lice(S', p) = -1 \).
    This is the teacher.
  \item The primal witness check asks to find \( p \) such that \( \Lice(S', p) = 1 \).
    This is the learner.
\end{itemize}
So the variant of \cref{alg:lagrangian:procedure} that monotonically updates \( \alpha \) is applicable to this setting.
\Cref{alg:ICE} describes the resulting procedure (which, unlike \cref{alg:lagrangian:procedure}, starts from the dual side).

\begin{algorithm}[t]
  \caption{~~ICE in the standard description (left) and as a primal-dual method (right)
  }\label{alg:ICE}
  \begin{minipage}{.47\linewidth}
    \begin{algorithmic}[1]
      \State let $ S' \leftarrow (\emptyset, \emptyset, \emptyset) $
      \While{$\textbf{true}$}
        \If {Safety of $ S' $ is unprovable}
          \State $ \mathbf{return}\;(\mathtt{unknown}, S') $
        \EndIf
        \State let $ p $ be a safety proof for $S'$
        \If {$ p $ is a safety proof for $\TSys$}
          \State $ \mathbf{return}\;(\mathtt{safe}, p) $
        \EndIf
        \State let $ S'' $ be counterexamples for $p$
        \State $ S' \leftarrow S' \cup S'' $
      \EndWhile
    \end{algorithmic}
  \end{minipage}
  \begin{minipage}{.47\linewidth}
    \begin{algorithmic}[1]
      \State let $ S' \leftarrow (\emptyset, \emptyset, \emptyset) $
      \While{$\textbf{true}$}
        \If {$ \sup_p \Lice(S', p) = -1 $}
          \State $ \mathbf{return}\;(\mathtt{unknown}, S') $
        \EndIf
        \State let $ p \in \{ p \in Y \mid  \Lice(S', p) = 1 \} $
        \If {$ \inf_{S''} \Lice(S'', p) = 1 $}
          \State $ \mathbf{return}\;(\mathtt{safe}, p) $
        \EndIf
        \State let $ S'' \in \{ S'' \in X \mid \Lice(S'', p) = -1 \} $
        \State $ S' \leftarrow S' \cup S'' $
      \EndWhile
    \end{algorithmic}
  \end{minipage}
\end{algorithm}

\section{Lagrangian for Primal-Dual Houdini}\label{sec:houdini}
This section analyzes \emph{primal-dual Houdini}~\cite{Padon2022}, a procedure that motivated this paper, in terms of the Lagrange duality.
The Lagrangian for primal-dual Houdini yields monotone progress on both sides of the duality;
it is symmetric in the sense that the primal and dual problems have the same high-level structure and can be swapped, changing the sign of the Lagrangian;
and interestingly showing that the Lagrangian is well-defined requires a non-trivial lemma.

\subsection{Warmup: CEGAR for Cartesian Abstraction}
\label{sec:ccegar}

As a stepping stone towards developing a Lagrangian for primal-dual Houdini, we first consider a version of CEGAR that operates not on predicate abstraction as in \Cref{sec:cegar}, but on Cartesian abstraction~\cite{DBLP:conf/tacas/BallPR01}. Given a set of predicates \( P \), the \emph{Cartesian} or \emph{conjunctive} abstraction it induces considers inductive invariants that can be expressed as a conjunction of predicates from \( P \). This abstraction is coarser than the Boolean or predicate abstraction, which considers invariants that can be expressed as an arbitrary Boolean formula over predicates from \( P \) (i.e., using conjunction, disjunction, and negation). While in the case of predicate abstraction the abstract counterexample takes the form of an abstract trace (i.e., a sequence of transitions),
in the case of Cartesian abstraction the abstract counterexample can be thought of as a directed acyclic graph (DAG) of transitions. We abstract over the details of the DAG structure,
and focus just on the set of states that participate in the abstract counterexample.

As in \Cref{sec:safety}, let us fix a transition system \( \TSys \) (\Cref{def:transition-system}) and a predicate set \( \PredSet \) (\Cref{def:predicate-set-for-a-transition-system}), \ie~a set with a satisfaction relation \( ({\models}) \subseteq \stateSet[\TSys] \times \PredSet  \). To define a Lagrangian for CEGAR for Cartesian abstraction, we let

\begin{equation*}
  X = \finitepowerset(\stateSet[\TSys])
  \quad\mbox{and}\quad
  Y = \finitepowerset(\PredSet).
\end{equation*}

We now define the value of \( \Lccegar(x,y) \) according to whether there is a conjunctive inductive invariant over predicates from \( y \) that proves that \( \proj{\TSys}{x} \) is safe, where \( \proj{\TSys}{x} \) is the transition system \(\TSys\) \emph{restricted} to states from \(x\), i.e.,
\( \stateSet[\proj{\TSys}{x}] = x \),
\( \initState[\proj{\TSys}{x}] = \initState[\TSys] \cap x \),
\( \trans[\proj{\TSys}{x}] = \trans[\TSys] \cap\, (x \times x) \), and
\( \badState[\proj{\TSys}{x}] = \badState[\TSys] \cap x \).
We define the Lagrangian \( \Lccegar \colon X \times Y \longrightarrow \{ -1, +1 \} \).
\begin{equation}
  \Lccegar(x,y) = \begin{cases}
    -1 & \mbox{if no \( \vartheta \subseteq y \) is a safe inductive invariant of \( \proj{\TSys}{x} \)} \\
    +1 & \mbox{otherwise}
  \end{cases}
  \label{eq:lagrangian-for-ccegar}
\end{equation}

Given a finite set of predicates \(y \in Y \), checking whether \(\inf_x \Lccegar(x,y) \geq 0\) (i.e., the dual witness check) amounts to checking whether \(\TSys\) can be proven safe using the Cartesian abstraction defined by \(y\). The well-known Houdini algorithm~\cite{Houdini,DBLP:journals/ipl/FlanaganJL01} is a procedure for answering this question, and it operates by a straightforward fixpoint computation to compute the strongest inductive invariant of \( \TSys \) that can be expressed as a conjunction of predicates from \( y \).
The algorithm also computes a set of (at most \(|y|\)) transitions that show why no larger subset of \(y\) is an inductive invariant for \( \TSys \);
or, in our terms, it produces a set of states \(x\) such that no larger subset of \(y\) is an inductive invariant for \( \proj{\TSys}{x} \).

Given a finite set of states \(x\), the primal witness check, i.e., checking if  \(\sup_y \Lccegar(x,y) \leq 0\), corresponds to finding the right predicates that can be used to prove safety of \(\proj{\TSys}{x}\) using Cartesian abstraction. That is, it corresponds to abstraction refinement, similarly to the primal witness check problem in \Cref{sec:cegar}.

\subsection{Induction Duality}

Let \( \TSys \) be a transition system
and \( \PredSet \) be a predicate set.
For a subset \( \vartheta \subseteq \PredSet \) of predicates, \( s \models \vartheta \) means \( \forall p \in \vartheta.\:s \models p \).
Similarly, for a subset \( \varpi \subseteq \stateSet[\TSys] \) of states, \( \varpi \models p \) means \( \forall s \in \varpi.\: s \models p \).
Finally, \( \varpi \models \vartheta \) means \( s \models p \) for every \( s \in \varpi \) and \( p \in \vartheta \).
Below we sometimes abuse the notation and use $p$ either for a single predicate or for a finite set of predicates.
We extend the definition of an inductive invariant (\Cref{def:predicate-set-for-a-transition-system}) to finite sets of predicates by interpreting the sets conjunctively. That is, a finite set of predicates \( \vartheta \subseteq \PredSet \) is an inductive invariant of if its conjunction, denoted \( \bigwedge \vartheta \), is an inductive invariant.

Primal-dual Houdini operates by analyzing \( \TSys \) together with another transition system \( \TSysI = (\stateSet[\TSysI], \initState[\TSysI], \trans[\TSysI], \badState[\TSysI]) \) which is defined over finite sets of predicates,
i.e., \( \stateSet[\TSysI] = \finitepowerset(\PredSet) \).
\( \TSysI \) is assumed to be an \emph{induction dual} of \( \TSys \), meaning that the following conditions are satisfied:%
\footnote{This presentation of induction duality is a slight adaptation of~\cite{Padon2022}. For example, the original paper assumes that the only element in \( \initState[\TSysI] \) is the empty set (representing \( \top \)), which (ID2) generalizes; it also assumes that \( \badState[\TSys] \) is represented by a single predicate \( p_0 \in \PredSet \), \ie~\( \badState[\TSys] = \{ s \in \stateSet[\TSys] \mid s \not\models p_0 \} \), and that \( \badState[\TSysI] = \{ y \in \finitepowerset(\PredSet) \mid p_0 \in y \} \), which (ID3) generalizes.
The original paper also assumes that transitions in \(\TSysI\) only grow the set of predicates, which we do not assume here.
  }
\begin{enumerate}
  \item[(ID1)] \( s_0 \models p \) for every \( s_0 \in \initState[\TSys] \) and \( p \in \stateSet[\TSysI] \), that is, the initial states satisfy all predicates; \label{it:id1}
  \item[(ID2)] \( s \models p_0 \) for every \( s \in \stateSet[\TSys] \) and \( p_0 \in \initState[\TSysI] \), that is, the initial predicates 
      are satisfied by all states; \label{it:id2}
  \item[(ID3)] \( s \not\models p \) for every \( s \in \badState[\TSys] \) and \( p \in \badState[\TSysI] \),
  that is, every bad state is excluded by every ``bad'' predicate set; and \label{it:id3}
  \item[(ID4)] if \( (s \trans[\TSys] s') \wedge (p \trans[\TSysI] p') \wedge (s \models p) \wedge (s \models p') \wedge (s' \models p) \), then \( (s' \models p') \), that is, transitions in \( \trans[\TSys] \) and in \( \trans[\TSysI] \) restrict each other in a particular way. \label{it:id4}
\end{enumerate}

The intuition behind this definition is that traces in \(\TSysI\) represents incremental induction proofs over \(\TSys\).
Specifically, (ID4) means that if we assume \(p\) is an invariant (of \(\TSys\)) and \( p \trans[\TSysI] p' \), then \(p'\) is also an invariant.
As the following lemma shows, a path in \( \TSysI \) implies an inductive invariant of \( \TSys \).

\begin{lemma}[\cite{Padon2022}]\label{lem:houdini:basic-lemma}
  Assume induction-dual transition systems \( \TSys \) and \( \TSysI \).
  If \( \initState[\TSysI] \ni p_0 \trans[\TSysI] p_1 \trans[\TSysI] \dots \trans[\TSysI] p_n \),
  then \( \bigcup_{i=0}^{n} p_i \) is an inductive invariant for \( \TSys \).
\end{lemma}
\begin{proof}
  Initiation follows from (ID1).
  For consecution, suppose
  \( \initState[\TSysI] \ni p_0 \trans[\TSysI] p_1 \trans[\TSysI] \dots \trans[\TSysI] p_n \)
  and assume \( s \models p_0 \cup p_1 \cup \cdots \cup p_n \) and \( s \trans[\TSys] s' \).
  We prove \( s' \models p_i \) for every \( 0 \leq i \leq n \) by induction on \( i \).
  For \( i = 0 \), the claim follows from (ID2).
  For \( i > 0 \), since \( s \models p_{i-1} \cup p_i \) and \( s' \models p_{i-1} \),
  we conclude \( s' \models p_i \) from (ID4)---the duality condition on transitions.
\end{proof}
\begin{corollary}
  For \( \TSys \) and \( \TSysI \) that are induction dual,
  if some \( p \in \badState[\TSysI] \) is reachable in \( \TSysI \) then \( \TSys \) is safe.
\end{corollary}
\begin{proof}
  Follows from \Cref{lem:houdini:basic-lemma} and (ID3).
\end{proof}
Dually, a path \( \initState[\TSys] \ni s_0 \trans[\TSys] s_1 \trans[\TSys] \dots \trans[\TSys] s_n \) in \( \TSys \) induces an inductive invariant for \( \TSysI \), and reachability to a bad state in \( \TSys \) implies safety of the dual transition system \( \TSysI \).

\subsection{Lagrangian for Primal-Dual Houdini}
\label{sec:pdh:lagrangian}

We can understand primal-dual Houdini as a simultaneous CEGAR for Cartesian abstraction for both \(\TSys\) and \(\TSysI\),
where the dual witness check in CEGAR for \(\TSys\) also solves the primal witness check for \(\TSysI\),
and the dual witness check for \(\TSysI\) solves the primal witness check for \(\TSys\). The fact that this can be done requires a non-trivial lemma about the connection between \(\TSys\) and \(\TSysI\) (and their restrictions).

Formally, let
\begin{equation*}
  X = \finitepowerset(\stateSet[\TSys])
  \quad\mbox{and}\quad
  Y = \finitepowerset(\PredSet).
\end{equation*}
The Lagrangian \( \Lpdh \colon X \times Y \longrightarrow \{ -1, 0, 1 \} \) is defined by
\begin{equation}
  \Lpdh(x,y) := \begin{cases}
    -1 & \mbox{if no \( \vartheta \subseteq y \) is a safe inductive invariant of \( \proj{\TSys}{x} \)} \\
    1 & \mbox{if no \( \varpi \subseteq x \) is a safe inductive invariant of \( \proj{\TSysI}{y} \)} \\
    0 & \mbox{otherwise}
  \end{cases}
  \label{eq:lagrangian-for-houdini}
\end{equation}
where \( \proj{\TSys}{x} \) is the subsystem of \( \TSys \) consisting of states in \( x \) and
\(\proj{\TSysI}{y}\) is the subsystem of \( \TSys \) consisting of sets of predicates from \( y \) (i.e., more formally this can be written as
\(\proj{\TSysI}{\powerset(y)}\)).
We need a slightly advanced result of the dual transition system to understand this definition, as well as to confirm its well-definedness.
The following theorem is new to this paper, an abstraction of the progress proof in \citet{Padon2022}.

\begin{theorem}\label{thm:houdini:well-definedness}
  Assume dual transition systems \( \TSys \) and \( \TSysI \).
  Then at least one of the following holds:
  \begin{itemize}
    \item some \( \vartheta \subseteq \PredSet \) is a safe inductive invariant for \( \TSys \), or
    \item some \( \varpi \subseteq \stateSet[\TSys] \) is a safe inductive invariant for \( \TSysI \).
  \end{itemize}
\end{theorem}
\begin{proof}
  \ifappendix
  See \Cref{sec:appx:proof-houdini}.
  \else
  See~\cite[Appendix~A]{popl25extended}.
  \fi
\end{proof}

\Cref{thm:houdini:well-definedness} shows that the conditions for \( \Lpdh(x,y) = 1 \) and for \( \Lpdh(x,y) = -1 \) do not hold simultaneously.
So \( \Lpdh \) given by \cref{eq:lagrangian-for-houdini} is well-defined.
Both \( X \) and \( Y \) are join semilattices, and \( L \) is anti-monotone in \( X \) and monotone in \( Y \).
So the general method applies to this Lagrangian, yielding a procedure in \cref{alg:Houdini} that enjoys the progress property.
The resulting procedure is the core of primal-dual Houdini in \cite{Padon2022}.
\begin{proposition}
  \( \sup_y \inf_x \Lpdh(x,y) \ge 0 \) implies the satefy of \( \TSys \).
\end{proposition}
\begin{proof}
  If \( \TSys \) is unsafe, there exists an error trace \( \initState[\TSys] \ni s_0 \trans[\TSys] \dots \trans[\TSys] s_n \in \badState[\TSys] \).
  Letting \( x = \{ s_0,\dots,s_n \} \), there is no safe inductive invariant of \( \TSys{\upharpoonright_x} \), so \( \sup_y L(x,y) = -1 \).
\end{proof}
Even if \( \sup_y \inf_x \Lpdh(x,y) = -1 \), we cannot conclude that \( \TSys \) is unsafe: it just shows that \( \PredSet \) is not sufficient to prove the safety of \( \TSys \).
Similar to the case of CEGAR and ICE, analysis of \( S_G \in X \) satisfying \( \sup_y \Lpdh(S_G, y) \le 0 \) may provide an unsafety proof.

The Lagrangian \( \Lpdh \) is symmetric in the sense that swapping \( \TSys \) and \( \TSysI \) does not essentially change the Lagrangian.
That means, writing \( \Lpdh^{\TSys, \TSysI} \) for the Lagrangian for the dual transition systems \( (\TSys, \TSysI) \), we have \( \Lpdh^{\TSys, \TSysI}(x,y) = -\Lpdh^{\TSysI, \TSys}(y,x) \).
By this symmetry, the primal and dual witness check problems are essentially the same.
Primal-dual Houdini~\cite{Padon2022} solves these subproblems by using the famous Houdini procedure~\cite{Houdini,DBLP:journals/ipl/FlanaganJL01}.

\Cref{thm:houdini:well-definedness} also suggests a connection to CEGAR.
By \cref{thm:houdini:well-definedness}, if \( \Lpdh(x,y) = 1 \) (\ie~no \( \varpi \subseteq x \) is a safe inductive invariant of \( \proj{\TSysI}{y} \)), then some \( \vartheta \subseteq y \) is a safe inductive invariant of \( \TSys \).
So \( \Lpdh(x,y) \ge 0 \) if and only if some \( \vartheta \subseteq y \) is a safe inductive invariant of \( \TSys \).
Therefore finding \( P \) such that \( \Lpdh(S_G,P) \ge 0 \) is essentially the abstraction refinement problem, which asks to find a predicate set that proves the safety of an approximation of the transition system \( \TSys \).
In primal-dual Houdini (\Cref{alg:Houdini}), line~10 computes \( P \) that satisfies \( \Lpdh(x,y) \ge 1 \) instead of \( \Lpdh(x,y) \ge 0 \),
and this is a requirement stronger than the standard abstraction refinement.
This strategy for updating the predicate set is the characteristic feature of primal-dual Houdini.

\begin{algorithm}[t]
  \caption{~~Primal-Dual Houdini}\label{alg:Houdini}
  \begin{algorithmic}[1]
    \Require ${\textsc{Primal-Dual Houdini}}$
    \State let $ S_G \leftarrow \emptyset $
    \State let $ P_G \leftarrow \{ p \} $ for some $ p \in \badState[\TSysI] $
      \While{$\textbf{true}$}
        \If {$ \inf_{S} \Lpdh(S, P_G) \ge 0 $}
          \State $ \mathbf{return}\;(\mathtt{safe}, P_G) $
        \EndIf
        \State let $ S \in \{ S \in X \mid \Lpdh(S, P_G) = -1 \} $
        \State $ S_G \leftarrow S_G \cup S $
        \If {$ \sup_{P} \Lpdh(S_G, P) \le 0 $}
          \State $ \mathbf{return}\;(\mathtt{unknown}, S_G) $
        \EndIf
        \State let $ P \in \{ P \in Y \mid  \Lpdh(S_G, P) = 1 \} $
        \State $ P_G \leftarrow P_G \cup P $
      \EndWhile
    \end{algorithmic}
\end{algorithm}

\subsection{A More Abstract View}

Our exposition above was rather close to that of~\cite{Padon2022}.
Specifically, we used sets of predicates for \(\TSysI\).
The use of sets of predicates is consistent with the perspective of incremental induction and constructing invariants in a conjunctive domain.
However, the Lagrangian formulation suggests that a more abstract structure is enough to define primal-dual Houdini, as follows. %

In the more abstract form, we assume transition systems \(\TSys\) and \(\TSysI\),
but do not assume anything on the state space \(\stateSet[\TSysI]\), i.e., we do not assume it consists of predicates or sets of predicates.
The only assumption is that there is a binary relation between
\(\stateSet[\TSys]\) and \(\stateSet[\TSysI]\), and that the two transition systems are induction duals, that is, conditions (ID1)--(ID4) hold. This is very similar to the induction-dual graphs of~\cite[Section 3.1]{Padon2022}.

Next, we assume arbitrary join semilattices \( X \) and \( Y \). That is, \( X \) is not
assumed to be finite sets of states, and \( Y \) is not assumed to be
finite sets of predicates.
Instead, we assume a \emph{restriction mapping},
which maps every element of \( X \) to a transition system \( \proj{\TSys}{x} \) that is a subsystem of \( \TSys \) in the following sense:
\( \stateSet[\proj{\TSys}{x}] \subseteq  \stateSet[\TSys] \),
\( \initState[\proj{\TSys}{x}] = \initState[\TSys] \cap \stateSet[\proj{\TSys}{x}]\),
\( \trans[\proj{\TSys}{x}]  = \trans[\TSys] \cap (\stateSet[\proj{\TSys}{x}] \times \stateSet[\proj{\TSys}{x}]) \), and
\( \badState[\proj{\TSys}{x}] = \badState[\TSys] \cap \stateSet[\proj{\TSys}{x}]\).
We further require that the restriction mapping is monotone,
i.e.,
if \( x \sqsubseteq_X x' \) then \( \stateSet[\proj{\TSys}{x}] \subseteq  \stateSet[\proj{\TSys}{x'}] \),
and that the restriction is \emph{unbounded} in the following sense:
\( \sup_x \stateSet[\proj{\TSys}{x}] = \stateSet[\TSys] \).
We assume a similar restriction mapping for $Y$ and $\TSysI$.

With this structure over the arbitrary sets $X$ and $Y$, we can define the Lagrangian by the following generalization of \cref{eq:lagrangian-for-houdini}:

\begin{equation}
  \Lpdh(x,y) = \begin{cases}
    -1 & \mbox{if no \( \vartheta \subseteq \stateSet[\proj{\TSysI}{y}] \) is a safe inductive invariant of \( \proj{\TSys}{x} \)} \\
    1 & \mbox{if no \( \varpi \subseteq \stateSet[\proj{\TSys}{x}] \) is a safe inductive invariant of \( \proj{\TSysI}{y} \)} \\
    0 & \mbox{otherwise}
  \end{cases}
  \label{eq:lagrangian-for-houdini-abstract}
\end{equation}

It is not too hard to see that \Cref{thm:houdini:well-definedness} generalizes to the setting \cref{eq:lagrangian-for-houdini-abstract} and that this Lagrangian is well-defined.
Note that the Lagrangian of \Cref{sec:pdh:lagrangian} is obtained from the more abstract one by letting
\( \stateSet[\TSysI] = \finitepowerset(\PredSet) \),
\( X = \finitepowerset(\stateSet[\TSys]) \),
\( Y = \finitepowerset(\PredSet) \),
and defining
\( \stateSet[\proj{\TSys}{x}] = x \), and
\( \stateSet[\proj{\TSysI}{y}] = \powerset(y) \)
(note the difference in the definition of the restriction mapping between \(X\) and \(Y\)).

\section{Lagrangians for Termination Verification}\label{sec:termination}
So far, we have focused on the safety verification problem.
This section discusses how our Lagrangian-based approach applies to termination verification, presenting ICE-based and CEGAR-based procedures.
Intuitively, the Lagrangians can be obtained by replacing predicates in the safety verification with ranking functions.
This idea works well for ICE, but for CEGAR, a problem arises with the monotonicity requirement.
From our perspective, the concept of \emph{disjunctive well-foundedness}~\cite{Podelski2004} is an idea to address the monotonicity requirement.

\subsection{Termination and Ranking Function}
Let \( \TSys = (\stateSet[\TSys], \initState[\TSys], \trans[\TSys]) \) be a transition system.
In this section, we omit the bad states \( \badState[\TSys] \) as this component is irrelevant to the verification problem in this section.
An \emph{infinite trace} is an infinite sequence \( s_0 s_1 s_2 \dots \in \stateSet[\TSys]^{\omega} \) of states such that \( s_0 \in \initState[\TSys] \) and \( s_i \trans[\TSys] s_{i+1} \) for every \( i \).
A transition system is \emph{terminating} if it has no infinite trace.

The termination of a transition system \( \TSys \) can be witnessed by a \emph{ranking function}.
In this paper, it is a function \( r \colon \stateSet[\TSys] \longrightarrow \kappa \) from states to an ordinal number \( \kappa \) such that \( \initState[\TSys] \ni s_0 \trans[\TSys]^* s \trans[\TSys] s' \) implies \( r(s) > r(s') \) (readers who are not familiar with ordinal numbers may set \( \kappa = \Nat \)).
So \( r \) must decrease for every reachable transition (but a transition \( s \trans[\TSys] s' \) from an unreachable state \( s \) does not require anything on \( r \)).
The existence of a ranking function is a sound and complete criterion for termination.

\subsection{ICE for Termination}
It is straightforward to modify the Lagrangian for the safety verification in \cref{sec:ice} to handle the termination verification.
Let \( \RankingFunc \subseteq (\stateSet[\TSys] \to \Nat) \) be a set of candidates for ranking functions, and let
\begin{equation*}
  X \::=\: \finitepowerset(\initState[\TSys]) \times \finitepowerset(\trans[\TSys])
  \quad\mbox{and}\quad
  Y \::=\: \RankingFunc.
\end{equation*}
The \( X \)-component is the same as the \( X \)-component for ICE (but the bad-state information is omitted).
The \( Y \)-component is the set of candidate proofs, which in the case of termination verification is the set of candidate ranking functions.
(Recall that in the case of safety verification \( Y \) is the set of predicates, thus in termination verification ranking functions assume the role of predicates in safety verification.)
The Lagrangian \( L_{\mathrm{T\mbox{-}ICE}} \colon X \times Y \longrightarrow \{ -1,1 \} \) is defined by
\begin{equation*}
  L_{\mathrm{T\mbox{-}ICE}}(S', r) =
  \begin{cases}
    \:\: 1 \quad & \mbox{\( r \) is a ranking function for the subsystem \( S' \) of \( \TSys \)} \\
    \:\: -1 & \mbox{otherwise.}
  \end{cases}
\end{equation*}
The \( X \)-component has an obvious join semilattice structure, and \( L_{\mathrm{T\mbox{-}ICE}} \) is anti-monotone on \( X \).
So the basic primal-dual procedure is applicable, provided that the primal and dual witness check problems are tractable.

We now discuss soundness.
An ``idealized'' version \( L'_{\mathrm{T\mbox{-}ICE}} \colon X \times Y' \longrightarrow \{-1,1\} \) of \( L_{\mathrm{T\mbox{-}ICE}} \) is given by setting \( X' := \powerset(\initState[\TSys]) \times \powerset(\trans[\TSys]) \) and \( Y' := (\stateSet[\TSys] \to \kappa) \) for a sufficiently large ordinal number \( \kappa \).\footnote{%
  One can choose \( \kappa \) as the minimum ordinal that has the same cardinality as \( \powerset(\stateSet[\TSys]) \).}
We define \( L'_{\mathrm{T\mbox{-}ICE}}(S', r) = 1 \) if and only if \( r \) is a ranking function
for the subsystem \( S' \) of \( \TSys \).
\begin{proposition}
  \( L'_{\mathrm{T\mbox{-}ICE}} \) enjoys strong duality, and \( \inf_x \sup_y L'_{\mathrm{T\mbox{-}ICE}}(x,y) = 1 \) iff\/ \( \TSys \) terminates.
\end{proposition}
\begin{proof}
  If \( \TSys \) terminates, a ranking function \( r \) for \( \TSys \) is given by
  \begin{align*}
    r(s) &\quad=\quad (\mbox{The number of steps remaining before termination}) \\
    &\quad=\quad \sup\,\{ r(s') + 1 \mid s \trans[\TSys] s' \}.
  \end{align*}
  This is a definition by induction on \( \trans[\TSys] \).
  Then \( L'_{\mathrm{T\mbox{-}ICE}}(S', r) = 1 \) since \( S' \) is a subsystem of \( \TSys \).
  If \( \TSys \) is not terminating, \( L'_{\mathrm{T\mbox{-}ICE}}(\TSys, r) = -1 \) for every \( r \in Y' \).
\end{proof}

\begin{lemma}\label{lem:termination:ice:soundness}
  \( \inf_{S'} L_{\mathrm{T\mbox{-}ICE}}(S', r) = 1 \) implies the termination of\/ \( \TSys \).
  \qed
\end{lemma}
\begin{proof}
  If \( r \) is not a ranking function of \( \TSys \), there exists a transition \( \initState[\TSys] \ni s_0 \trans[\TSys] \dots \trans[\TSys] s_n \trans[\TSys] s_{n+1} \) such that \( r(s_n) \not> r(s_{n+1}) \).
  Then the subsystem \( S' = (\{ s_0 \}, \{ (s_i,s_{i+1}) \mid 0 \le i \le n\}) \) satisfies \( L'_{\mathrm{T\mbox{-}ICE}}(S', r) = -1 \).
\end{proof}

We discuss the tractability of the primal and dual witness check.

The dual witness check is tractable because whether \( \inf_{S'} L_{\mathrm{T\mbox{-}ICE}}(S', r) = -1 \) is reducible to a safety verification problem.
Let \( \TSys^{(r)} \) be the transition system (with bad states) given by \( \stateSet[\TSys^{(r)}] := \stateSet[\TSys] \times \stateSet[\TSys] \), \( \initState[\TSys^{(r)}] := \{ (s, s') \mid s \in \initState[\TSys], s \trans[\TSys] s' \} \), \( ({\trans[\TSys^{(r)}]}) := \{ ((s, s'), (s', s'')) \mid s' \trans[\TSys] s'' \} \) and \( \badState[\TSys^{(r)}] := \{ (s, s') \mid r(s) \le r(s') \} \).
A state of \( \TSys^{(r)} \) is a pair \( (s, s') \) of a previous state \( s \) and a current state \( s' \), and a pair is bad if and only if it violates the ranking condition with respect to \( r \).
If the system \( \TSys^{(r)} \) is safe, \( r \) is a ranking function for \( \TSys \).
The safety verification of \( \TSys^{(r)} \) can be solved by procedures
such as those discussed in \cref{sec:safety,sec:houdini}.
If \( \TSys^{(r)} \) is unsafe, an error trace \( (s_0, s_1) (s_1, s_2) \dots (s_n, s_{n+1}) \in \stateSet[\TSys^{(r)}]^* \) induces \( S' = (\{ s_0 \}, \{ (s_0, s_1) (s_1, s_2) \dots (s_n, s_{n+1}) \}) \).

On the contrary, it is hard to check the primal witness, which asks to find a ranking function \( r \) that works well on the subsystem \( S' \).
Since \( S' \) is a terminating finite transition system, a ranking function can in theory be given by mapping each point to the number of remaining steps, but this would not be expected to behave well for points outside of \( S' \).
To find generalized ranking functions, methods have been proposed that use template-based synthesis~\cite{Unno2021} as well as machine learning techniques, such as decision tree learning~\cite{Kura2021} and support vector machines~\cite{Li2020}.

\subsection{CEGAR for Termination}\label{sec:termination:dwf}
The Lagrangian for CEGAR has a semilattice structure on \( Y \) and \( \Lcegar \) is monotone on \( Y \).
In the context of termination verification, \( Y \) consists of (sets of) candidate ranking functions, but it is not straightforward to introduce the join semilattice structure to the \( Y \)-component.

A \emph{disjunctively well-founded relation}~\cite{Podelski2004} is a concept to address this issue: a relation \( \succ \) is \emph{disjunctively well-founded} if it is a finite union of well-founded relations.
For a finite set \( R \in \finitepowerset(\RankingFunc) \) of candidate ranking functions, the relation \( \succ_R \) defined by \( (s \succ_R s') :\Leftrightarrow \exists r \in R. r(s) > r(s') \) is a disjunctively well-founded relation.
The set \( \finitepowerset(\RankingFunc) \) has an obvious lattice structure, and \( R \subseteq R' \in \finitepowerset(\RankingFunc) \) implies \( ({\succ_R}) \subseteq ({\succ_{R'}}) \).

Let
\(
  X \::=\: \{ s_0 s_1 \dots s_n \in \stateSet[\TSys] \mid \initState[\TSys] \ni s_0 \trans[\TSys] \dots \trans[\TSys] s_n \}
\),
\( Y \::=\: \finitepowerset(\RankingFunc) \),
and
\begin{equation*}
  L_{\mathrm{T\mbox{-}CEGAR}}(s_0 \dots s_n, R) =
  \begin{cases}
    \:\: 1 \quad & \mbox{\( \forall i. \forall j. (i < j) \Rightarrow (s_i \succ_R s_j) \)} \\
    \:\: -1 & \mbox{otherwise.}
  \end{cases}
\end{equation*}
Note that \( \succ_R \) must satisfy \( s_i \succ_R s_j \) for every \( i < j \), including the case that \( i+1 \neq j \), \ie~\( s_j \) is not the immediate successor of \( s_i \), unlike the condition for ranking functions.
Then \( Y \) is a join semilattice and \( L_{\mathrm{T\mbox{-}CEGAR}} \) is monotone on \( Y \).
When \( \sup_R \inf_\tau L_{\mathrm{T\mbox{-}CEGAR}}(\tau, R) = 1 \), then \( \TSys \) is terminating by the following theorem.
\begin{theorem}[{\citet[Theorem~1]{Podelski2004}}]
  A transition system \( \TSys \) is terminating if there exists a disjunctively well-founded relation \( \succ \) such that \( \initState[\TSys] \ni s_0 \trans[\TSys]^* s \trans[\TSys]^+ s' \) implies \( s \succ s' \).
  \qed
\end{theorem}
\begin{corollary}
  If \( \sup_y \inf_x L_{\mathrm{T\mbox{-}CEGAR}}(x,y) = 1 \), then \( \TSys \) is terminating.
  \qed
\end{corollary}

The dual witness check, which asks to find \( \tau \) such that \( L_{\mathrm{T\mbox{-}ICE}}(\tau, R) = -1 \), is reducible to a safety verification problem.
The idea is similar to the above case, but the first component \( s \) of a pair \( (s, s') \) is now a past state that is not necessarily the previous state.
Formally, let \( \TSys^{(R)} \) be the transition system (with bad states) given by \( \stateSet[\TSys^{(R)}] := \stateSet[\TSys] \times \stateSet[\TSys] \), \( \initState[\TSys^{(R)}] := \{ (s, s') \mid s \in \initState[\TSys], s \trans[\TSys] s' \} \), \( ({\trans[\TSys^{(R)}]}) := \{ ((s_1, s_1'), (s_2, s_2')) \mid s_1' \trans[\TSys] s_2' \mbox{ and } s_2 = s_1 \vee s_2 = s_1' \} \) and \( \badState[\TSys^{(R)}] := \{ (s, s') \mid s \not\succ_R s' \} \).
For an error trace \( (s_0, s_0') (s_1, s_1') \dots (s_n, s_n') \in \stateSet[\TSys^{(R)}]^* \) for \( \TSys^{(R)} \), we have \( \initState[\TSys] \ni s_0 \trans[\TSys] s_0' \trans[\TSys] s_1' \trans[\TSys] \dots \trans[\TSys] s_n' \), \( s_n \not\succ_R s_n' \) and \( s_n \in \{ s_0, s_0', s_1', \dots, s_{n-1}' \} \).

The primal witness check is more tractable than \( L_{\mathrm{T\mbox{-}ICE}} \) since it suffices to find a ranking function that works for a single trace.

The basic primal-dual procedure applied to \( L_{\mathrm{T\mbox{-}CEGAR}} \) is (the core of) the procedure known as \textsc{Terminator}~\cite{Cook2005,Cook2006}.

\section{Lagrangian for Quantified Linear Arithmetic Solver}\label{sec:qlra}
This section deals with a problem that is closely related to verification, but different in nature from the safety/termination verification problems we have dealt with so far.
The problem is validity checking of first-order predicate logic formulas, and this section analyses the procedure given by \citet{Farzan2016} from the viewpoint of Lagrange duality.
Basically, the Lagrangian corresponds to Skolemization, but as discussed in \citet{Farzan2016}, Skolemization alone is not sufficient.
Interestingly, our framework of Lagrangian duality and basic primal-dual procedure clarifies the issue: as we shall see, it is the monotonicity criterion.

\subsection{First-Order Formula and Skolemization}\label{sec:qlra:skolem}
This section focuses on first-order formulas.
We consider the theory of linear rational arithmetic, \textsf{LRA}, which is the theory that \citet{Farzan2016} mainly dealt with.
A term is given by \( t ::= x \mid c \mid t_1 + t_2 \mid c \cdot t \) (where \( c \in \mathbb{Q} \)) and atomic predicates are \( t_1 < t_2 \) and \( t_1 \le t_2 \).  We assume that formulas are in prenex normal form, \ie{}, \( \varphi = \mathcal{Q}_1 x_1. \dots \mathcal{Q}_k x_k. \vartheta \) where \( \mathcal{Q}_i \in \{ \forall, \exists \} \) and \( \vartheta \) is quantifier-free.
Since \textsf{LRA} enjoys effective quantifier elimination (\ie, given a quantified formula \( \varphi(x,\vec{y}) \), one can effectively construct a quantifier-free formula \( \psi(\vec{y}) \) such that \( \psi(\vec{y}) \leftrightarrow \exists x. \varphi(x, \vec{y}) \)), the validity checking problem for first-order logic formulas over \textsf{LRA} is decidable.

\emph{Skolemization} is a satisfiability-preserving translation of first-order logic formulas, introducing new functional symbols.
Consider a quantified formula:
\begin{equation}
  \forall a \in \Rational. \exists b \in \Rational. \forall c \in \Rational. \varphi(a,b,c).
  \label{eq:qlra:example-formula}
\end{equation}
If this formula is true, there exists a function \( f_b \colon \Rational \to \Rational \) that maps a value assigned to \( a \) to an appropriate value for \( b \), \ie, for every \( v \in \Rational \),
\begin{equation*}
  \forall c \in \Rational. \varphi(v, f_b(v), c),
\end{equation*}
and \( f_b \) is called a \emph{Skolem function} for \( b \).
By using the Skolem function, the validity of \cref{eq:qlra:example-formula} is reduced to the problem to find an appropriate function \( f_b \) such that
\begin{equation*}
  \forall a \in \Rational. \forall c \in \Rational. \varphi(a, f_b(a), c)
\end{equation*}
is valid, and this validity problem can be solved by SMT solvers (provided that \( f_b \) is describable as a term in \textsf{LRA}).
Conversely, if the formula in \cref{eq:qlra:example-formula} is false, its negation \( \exists a \in \Rational. \forall b \in \Rational. \exists c \in \Rational. \neg\varphi(a,b,c) \) is true, so there exist Skolem functions \( f_a \colon \mathbb{Z} \) and \( f_c \colon \mathbb{Z} \to \mathbb{Z} \) such that
\begin{equation*}
  \forall b \in \mathbb{Z}. \neg \varphi(f_a, b, f_c(b)).
\end{equation*}

The above observation yields the following Lagrangian:
letting
\begin{align*}
  X &= \mbox{(Skolem functions for \( a \) and \( c \))} = (\mathbb{Z} \times (\mathbb{Z} \to \mathbb{Z})) \\
  Y &= \mbox{(Skolem function for \( b \))} = (\mathbb{Z} \to \mathbb{Z}),
\end{align*}
the Lagrangian \( L \colon X \times Y \longrightarrow \{ -1,1 \} \) is given by
\begin{equation*}
  L((f_a, f_c), f_b) = 1
  \quad:\Leftrightarrow\quad
  \mbox{\( \varphi(f_a, f_b(f_a), f_c(f_b(f_a))) \) is true}.
\end{equation*}

The optimal value of the dual optimization problem is \( 1 \) if and only if there exists an appropriate Skolem function for \( b \).
To see this, note that
\begin{align*}
  \inf_{f_a,f_c} L((f_a, f_c), \beta) = -1
  \quad&\Leftrightarrow\quad
  \exists f_a \in \Rational. \exists f_c \in (\Rational \to \Rational). \varphi(f_a, \beta(f_a), f_c(\beta(f_a))) \mbox{ is false}
  \\
  &\Leftrightarrow\quad
  \exists a \in \Rational. \exists c \in \Rational. \varphi(a, \beta(a), c) \mbox{ is false},
\end{align*}
so, by negating the both sides,
\begin{equation*}
  \inf_{f_a,f_c} L((f_a, f_c), \beta) = 1
  \quad\Leftrightarrow\quad
  \forall a \in \Rational. \forall c \in \Rational. \varphi(a, \beta(a), c).
\end{equation*}
Hence, \( \beta \in Y \) witnessing \( \sup_{y \in Y} \inf_{x \in X} L(x,y) = 1 \) is a Skolem function for \( b \) that witnesses the truth of the formula in \cref{eq:qlra:example-formula}.
Dually, a witness \( \alpha \in X \) of \( \inf_{x \in X} \sup_{y \in Y} L(x,y) = -1 \) is a pair of Skolem functions for \( a \) and \( c \) that witnesses the falsity of the formula.

In the above setting, both \(X\) and \(Y\) are the sets of all possible Skolem functions in the semantic domains, so there is an appropriate choice for exactly one of \( X \) and \( Y \).
Therefore, we have the strong duality:
\begin{equation*}
  \sup_{f_b} \inf_{f_a,f_c} L((f_a,f_c), f_b)
  \quad=\quad
  \inf_{f_a,f_c} \sup_{f_b} L((f_a,f_c), f_b),
\end{equation*}
and the optimal value coincides with the truth of the formula.

This argument can be applied to arbitrary formulas.
Given a first-order predicate logic formula \(\psi\) with quantifiers, there exists a Lagrangian \( L_{\psi} \colon X_\psi \times Y_\psi \longrightarrow \{-1,1\} \), where \( X_\psi \) (resp.~\( Y_\psi \)) is the set of Skolem functions for universally (resp.~existentially) quantified variables in \( \psi \), such that \( L_\psi \) satisfies the strong duality and the optimal value of \( L_\psi \) is the truth of \( \psi \).

\subsection{Quantified LRA Solver by Farzan and Kincaid}
Very roughly, the procedure by \citet{Farzan2016} is an instance of the basic primal-dual procedure for Lagrangians in \cref{sec:qlra:skolem}.
However, we cannot directly apply the basic primal-dual procedure since the Lagrangian does not satisfy the monotonicicy criterion, \ie~the sets \( X \) and \( Y \) of Skolem functions are not poset (and the Lagrangian is not monotone).
Another, relatively minor issue is that the sets \( X \) and \( Y \) are too large, containing elements that have no finite representations.

\citet{Farzan2016} introduced the notion of \emph{strategy skeletons} to address the above issues.
Intuitively, a strategy skeleton is a variant of Skolem functions, albeit with the additional capacity to select multiple values, claiming that at least one of the selected values is appropriate.
The inclusion of the set of selected values is an order for strategy skeletons, and the Lagrangian is monotone (or antimonotone) with respect to this order.
\begin{definition}[Strategy Skeleton]
  The set of \emph{\textsc{SAT} strategy skeletons} is defined by the grammar
  \begin{equation*}
    \textstyle
    \pi \quad::=\quad
    \bullet \mid \forall y. \pi \mid \bigsqcup_{i=1,\dots,n} t_i.\pi_i
  \end{equation*}
  where \( t_i \) is a term of \textsf{LRA} and \( t_i \neq t_j \) if \( i \neq j \).
  For a formula $\psi$ in prenex normal form
  we write \( \vec{x} \vdash \psi \lhd \pi \) to denote that \( \pi \) is a \textsc{SAT} strategy skeleton for $\psi$
  and satisfies \( \mathrm{fv}(\pi) \subseteq \{\vec{x}\}\).
  Formally, it is the relation defined by the following rules:
  \begin{equation*}
    \dfrac{\mathit{QF}(\vartheta)}{
      \vec{x} \vdash \vartheta \lhd \bullet
    }
    \quad
    \dfrac{\vec{x}, y \vdash \psi \lhd \pi}{
      \vec{x} \vdash (\forall y.\psi) \lhd (\forall y.\pi)
    }
    \quad
    \dfrac{\vec{x} \vdash \psi \lhd \pi_i \mbox{ for every \( i \)} \qquad \forall i. \mathrm{fv}(t_i) \subseteq \{\vec{x}\}}{
      \vec{x} \vdash (\exists y.\psi) \lhd (\bigsqcup_i t_i.\pi_i)
    }
  \end{equation*}
  The strategy skeleton \( \bullet \) is for quantifier-free formulas, to which we have nothing to choose.
  The strategy skeleton \( \forall y.\pi \) is for universally-quantified formulas \( \forall y.\psi \): a \textsf{SAT} strategy skeleton does nothing on \( \forall y \), and \( \pi \) describes a way to choose values for \( \exists \)-quantified variables in \( \psi \).
  The strategy skeleton \( \bigsqcup_{i=1,\dots,n} t_i.\pi_i \) is for \( \exists y.\psi \), meaning that an appropriate choice for \( y \) should be found in \( \{t_1,\dots,t_n \} \).
  An \textsf{UNSAT} strategy skeleton can be defined similarly:
  \begin{gather*}
    \textstyle
    \varrho \quad::=\quad
    \bullet \mid \exists y. \varrho \mid \bigsqcap_{i = 1,\dots,n} t_i.\varrho_i
    \\
    \dfrac{\mathit{QF}(\vartheta)}{
      \vec{x} \vdash \bullet \rhd \vartheta
    }
    \quad
    \dfrac{\vec{x}, y \vdash \varrho \rhd \psi}{
      \vec{x} \vdash (\exists y. \varrho) \rhd (\exists y. \psi)
    }
    \quad
    \dfrac{\vec{x} \vdash \varrho_i \rhd \psi \mbox{ for every \( i \)} \qquad \forall i. \mathrm{fv}(t_i) \subseteq \{\vec{x}\}}{
      \vec{x} \vdash (\bigsqcap_i t_i. \varrho_i) \rhd (\forall y.\psi)
    }.
  \end{gather*}
  For a sentence \( \varphi \) in prenex normal form, let \( \SATSkeleton(\varphi) \) and \( \UNSATSkeleton(\varphi) \) be the set of \textsf{SAT} and \textsf{UNSAT} strategy sketelons for \( \varphi \) with no free variable, \ie~\( \SATSkeleton(\varphi) := \{ \pi \mid {}\vdash \varphi \lhd \pi \} \) and \( \UNSATSkeleton(\varphi) := \{ \varrho \mid {}\vdash \varrho \rhd \varphi \} \).
  \qed
\end{definition}

\begin{example}[{\cite{Farzan2016}}]\label{eg:qlra:example}
  Let
  \begin{equation*}
    \varphi
    \quad:=\quad
    \exists w. \forall x. \exists y. \forall z. \vartheta(w,x,y,z),
    \qquad \vartheta(w,x,y,z) \:\equiv\: ((y < 1 \vee 2w < y) \wedge (z < y \vee x < z)).
  \end{equation*}
  A \textsf{SAT} strategy skeleton is
  \begin{equation*}
    \pi
    \quad:=\quad
    0 \,.\, \forall x \,.\, \big((x \,.\, \forall z \,.\, \bullet) \sqcup (2x \,.\, \forall z \,.\, \bullet) \big).
  \end{equation*}
  This skeleton selects \( 0 \)
  for \( w \) and \( x \)
  or \( 2x \)
  for \( y \) depending on the branch of \( \sqcup \).
  \qed
\end{example}

The set \( \SATSkeleton(\varphi) \) is a preordered set by the order defined by the following rules:
\begin{equation*}
  \dfrac{\mathstrut}{\bullet \le \bullet}
  \quad
  \dfrac{\pi \le \pi'}{\forall y.\pi \le \forall y.\pi'}
  \quad
  \dfrac{\forall i. \exists j. t_i = t'_j \wedge \pi_i \le \pi'_j}{\bigsqcup_i t_i.\pi_i \le \bigsqcup_j t'_j.\pi'_j}
\end{equation*}
The join operation \( \cup \) is defined by \(\bullet \cup \bullet := \bullet\), \( (\forall y. \pi) \cup (\forall y. \pi') := \exists y. (\pi \cup \pi')
 \) and
\begin{equation*}
  \textstyle
  (t.\pi) \cup (\bigsqcup_{i \in I} t_i.\pi_i) :=
  \begin{cases}
    (t.\pi) \sqcup (\bigsqcup_{i \in I} t_i.\pi_i) & \quad\mbox{if \( t \neq t_i \) for every \( i \)} \\
    (t. (\pi \cup \pi_j)) \sqcup (\bigsqcup_{i \in (I \setminus \{j\})} t_i.\pi_i) & \quad\mbox{if \( t = t_j \)}.
  \end{cases}
\end{equation*}
Dually \( \UNSATSkeleton(\varphi) \) has a similar structure, obtained by replacing \( \sqcup \) with \( \sqcap \) (\eg~\( \varrho \le (\varrho \sqcap \varrho') \)).
We write the join in \( \UNSATSkeleton(\varphi) \) as \( \cup \), since it computes the union of candidates.

A pair \( (\varrho, \pi) \in \UNSATSkeleton(\varphi) \times \SATSkeleton(\varphi) \) of \textsf{SAT} and \textsf{UNSAT} skeletons for the same formula \( \varphi \) induces a quantifier-free formula \( \langle \varrho \,|\, \varphi \,|\, \pi \rangle \) defined by induction on \( \varphi \) as follows:
\begin{align*}
  \begin{array}{rclcl}
  \langle \bullet \,| & \hspace{-7pt} \vartheta \hspace{-7pt} & |\, \bullet \rangle
  &\quad:=\quad& \vartheta
  \\
  \langle \bigsqcap_i t_i.\varrho_i \,| & \hspace{-7pt} \forall x.\varphi \hspace{-7pt} & |\, \forall x.\pi \rangle
  &\quad:=\quad& \bigwedge_i \langle \varrho_i \,|\, \varphi[t_i/x] \,|\, \pi[t_i/x] \rangle
  \\
  \langle \exists x.\varrho \,| & \hspace{-7pt}\exists x.\varphi\hspace{-7pt} & |\, \bigsqcup_i t_i.\pi_i \rangle
  &\quad:=\quad& \bigvee_i \langle \varrho[t_i/x] \,|\, \varphi[t_i/x] \,|\, \pi_i \rangle.
  \end{array}
\end{align*}
Let \( L_{\mathrm{FK}} \colon \UNSATSkeleton(\varphi) \times \SATSkeleton(\varphi) \longrightarrow \{ -1,1 \} \) be the Lagrangian defined by
\begin{equation*}
  L_{\mathrm{FK}}(\varrho, \pi) = 1
  \quad:\Leftrightarrow\quad
  \langle \varrho \,|\, \varphi \,|\, \pi \rangle \mbox{ is true}.
\end{equation*}
This Lagrangian is written as \( L^\varphi_{\mathrm{FK}} \) if the formula \( \varphi \) should be clarified.
The procedure by \citet{Farzan2016} is an instance of the basic primal-dual procedure for \( L_{\mathrm{FK}} \).

\begin{example}[{\cite{Farzan2016}}]
  Recall the formulas \( \varphi \) and \( \vartheta \) and the \textsc{SAT} strategy skeleton \( \pi \) in \cref{eg:qlra:example}.
  An example of an \textsf{UNSAT} strategy skeleton is
  \begin{equation*}
    \varrho
    \quad:=\quad
    \exists w \,.\, (-1) \,.\, \exists y \,.\, \big((u_1 \,.\, \bullet) \sqcap (u_2 \,.\, \bullet) \big),
    \qquad
    u_1(w,y) = y,
    \quad
    u_2(w,y) = (w+y)/2
  \end{equation*}
  Then, the calculation of \( \langle \varrho \,|\, \varphi \,|\, \pi \rangle \) proceeds as
  \begin{align*}
    &
    \Big\langle \exists y. ((u_1(0, y) \,.\, \bullet) \sqcap (u_2(0,y) \,.\, \bullet))\,\Big|\, \exists y. \forall z. \vartheta(0, -1, y, z) \,\Big|\, ((-1) \,.\, \forall z \,.\, \bullet) \sqcup ((-2) \,.\, \forall z \,.\, \bullet) \Big\rangle
    \\
    =\: &
    \Big\langle ((u_1(0, -1) \,.\, \bullet) \sqcap (u_2(0,-1) \,.\, \bullet))\,\Big|\, \forall z. \vartheta(0, -1, -1, z) \,\Big|\, \forall z \,.\, \bullet \Big\rangle
    \\
    & \quad \vee \Big\langle ((u_1(0, -2) \,.\, \bullet) \sqcap (u_2(0,-2) \,.\, \bullet))\,\Big|\, \forall z. \vartheta(0, -1, -2, z) \,\Big|\, \forall z \,.\, \bullet \Big\rangle
    \\
    =\: &
    \Big(\vartheta\big(0, -1, -1, u_1(0, -1)\big) \ \wedge\  \vartheta\big(0, -1, -1, u_2(0, -1)\big) \Big)
    \\
    & \quad \vee \Big( \vartheta\big(0, -1, -2, u_1(0, -2)\big) \ \wedge\ \vartheta\big(0, -1, -2, u_2(0, -2)\big) \Big)
  \end{align*}
  so we have \( L_{\mathrm{FK}}(\varrho,\pi) = \langle \varrho \,|\, \varphi \,|\, \pi \rangle = \big(\bot \wedge \top\big) \vee \big(\bot \wedge \bot\big) = -1 \).
  \qed
\end{example}

\begin{lemma}
  \( L_{\mathrm{FK}} \colon \UNSATSkeleton(\varphi) \times \SATSkeleton(\varphi) \longrightarrow \{-1,1\} \) is anti-monotone on the first argument and monotone on the second argument.
\end{lemma}
\begin{proof}
  Intuitively, if \( \pi \le \pi' \), then \( \pi' \) has more components than \( \pi \) connected by \( \sqcup \).
  Then \( \langle \varrho | \varphi | \pi' \rangle \) has more components than \( \langle \varrho | \varphi | \pi \rangle \) connected by \( \vee \).
  So, the validity of the latter implies the validity of the former.
  The anti-monotonicity on \( \UNSATSkeleton(\varphi) \) is similar.
\end{proof}

Prior to the discussion of correctness, we discuss the decidability of the primal/dual witness check problems, following \citet{Farzan2016}.
We discuss the dual witness check here; the primal is essentially equivalent.
Assume a \textsf{SAT} strategy skeleton \( \pi \in \SATSkeleton(\varphi) \).
Then \( \inf_\varrho L^\varphi_{\mathrm{FK}}(\varrho, \pi) = 1 \) is reducible to the validity of the formula \( \varphi | \pi \rangle \) given as follows:
\begin{gather*}
  \textstyle
  \vartheta |\bullet \rangle := \vartheta
  \qquad
  (\forall x.\varphi) | \forall x.\pi \rangle := \forall x. (\varphi | \pi \rangle)
  \qquad
  (\exists x.\varphi) | \bigsqcup_i t_i.\pi_i \rangle := \bigvee_i (\varphi[t_i/x]) | \pi_i \rangle.
\end{gather*}
Since \( \varphi | \pi \rangle \) is a formula with no \( \exists \), its validity can be checked by an SMT solver.
\begin{example}\label{eg:qlra:smt-and-counterstrategy-extraction}
  Recall the formulas \( \varphi \) and \( \vartheta \) and a \textsf{SAT} strategy skeleton \( \pi \) in \cref{eg:qlra:example}.
  Then
  \begin{align*}
    \varphi|\pi\rangle
    \quad&=\quad
    \forall x. \big((\forall z_1. \vartheta(0,x,x,z_1)) \vee (\forall z_2. \vartheta(0,x,2x,z_2)) \big),
  \end{align*}
  which is invalid.
  For example, consider the assignment \( x = -1 \), \( z_1 = -1 \) and \( z_2 = -2 \).
  This assignment gives a \textsf{UNSAT} strategy skeleton \( \varrho := \forall w.(-1).\forall y. \big((-1).\bullet \,\sqcup\, (-2).\bullet\big) \) such that \( L_{\mathrm{FK}}(\varrho,\pi) = -1 \).
  For another \textsf{SAT} strategy skeleton \( \pi' := (-2).\forall x.(x+1).\forall z.\bullet \), we have
  \(
    \varphi|\pi'\rangle
    =
    \forall x. \forall z. \vartheta(-2, x, x+1, z)
  \),
  which is valid.
  \qed
\end{example}

\begin{lemma}\label{lem:qlra:smt}
  \( \inf_\varrho L^\varphi_{\mathrm{FK}}(\varrho, \pi) = 1 \) if and only if \( \varphi | \pi \rangle \) is valid.
\end{lemma}
\begin{proof}
\ifappendix
See \Cref{sec:appx:proof-of-smt-correctness}.
\else
See~\cite[Appendix~B]{popl25extended}.
\fi
\end{proof}

\begin{theorem}
  For every closed formula \( \varphi \) over \textsf{LRA}, the Lagrangian \( L^\varphi_{\mathrm{FK}} \) enjoys the strong duality, and its optimal value coincides with the validity of \( \varphi \).
\end{theorem}
\begin{proof}
  The result follows from \citet{Farzan2016}.
  To show the strong duality, it suffices to prove that, for a valid formula \( \varphi \), there exists a \textsf{SAT} strategy skeleton \( \pi \in \SATSkeleton(\varphi) \) such that \( \inf_{\varrho} L^\varphi_{\mathrm{FK}}(\varrho, \pi) = 1 \).
  The basic observation in \citet{Farzan2016} is that the quantifier elimination \( \vartheta(y) \leftrightarrow \exists x. \psi(x,y) \) can be achieved by substitution, \ie~there exists a finite set \( \{ t_1,\dots,t_n \} \) of terms such that \( (\exists x.\psi(x,y)) \leftrightarrow (\psi(t_1,y) \vee \dots \vee \psi(t_n,y)) \).\footnote{\citet{Farzan2016} was inspired by \emph{model-based projection}~\cite{Komuravelli2014}, which is closely related to quantifier elimination (see \cite[Section~4.1]{Farzan2016}).}
  Iterative application of this result yields a \textsf{SAT} strategy skeleton \( \pi \) such that \( \varphi \leftrightarrow (\varphi|\pi\rangle) \).

  Since \( \psi(t) \to (\exists x.\psi(x)) \) is valid for every \( \psi \) and \( t \), the formula \( \varphi|\pi\rangle \), obtained by instantiating \( \exists \)-variables, is stronger than \( \varphi \).
  By \cref{lem:qlra:smt}, \( \inf_\varrho L^{\varphi}_{\mathrm{FK}}(\varrho, \pi) = 1 \) implies the validity of \( \varphi|\pi\rangle \) and hence the validity of \( \varphi \).
\end{proof}

The algorithm given by \citet{Farzan2016} is an instance of the basic primal-dual procedure for the Lagrangian \( L_{\mathrm{FK}} \).
Given a \textsf{SAT} strategy skeleton \( \pi \), the check of \( \inf_\varrho L_{\mathrm{FK}}(\varrho,\pi) = 1 \) is reducible to the validity of \( \varphi|\pi\rangle \), which can de solved by an SMT solver.
If \( \inf_\varrho L_{\mathrm{FK}}(\varrho,\pi) = -1 \), or equivalently, \( \varphi|\pi\rangle \) is invalid, the SMT solver generates a counter-model for \( \varphi|\pi\rangle \), from which a \textsf{UNSAT} strategy skeleton \( \varrho \) can be constructed as illustrated in \cref{eg:qlra:smt-and-counterstrategy-extraction}.
But their procedure has an important twist here: the assignment to each variable in the counter-model, which is a concrete rational number, is converted to a term.
The terms are chosen from those playing an important role in quantifier elimination or model-based projection (\cf~\cite[Section~4.1]{Farzan2016}).
As the set of such terms is finite, \citet{Farzan2016} in effect considered finite subsets of \( \SATSkeleton(\varphi) \) and \( \UNSATSkeleton(\varphi) \), and this finiteness guarantees that Farzan and Kincaid's procedure is terminating.

\section{Lagrangian for Fixed-Point Logic over Quantified Linear Arithmetic}\label{sec:fixed-point-logic}
This section demonstrates the usefulness of our Lagrangian-based approach by developing a solver of the validity problem for a fixed-point logic over quantified linear arithmetic.
The validity problem is closely related to game solving against liveness winning criteria as in \citet{Heim2024}.

Our observation is that the difficulty of this problem can actually be decomposed into the difficulties of termination analysis and of quantifiers, each of which has been addressed in \cref{sec:termination,sec:qlra}.
So a Lagrangian for the validity problem of fixed-point logic with quantifiers is obtained as a simple combination of ideas in these sections.

We have developed a prototype implementation of the basic primal-dual procedure for the proposed Lagrangian and provide an evaluation.

\subsection{Fixed-Point Logic}
We define the syntax and semantics of the fixed-point logic studied in this section.
Although the discussion in this section is applicable to first-order structures in general to some extent, we will focus on the fixed-point logic over linear integer arithmetic, \textsf{LIA}.
An atomic predicate \( p \) is either \( = \) or \( \le \), and a term is given by \( t ::= n \mid x \mid t + t' \mid n \times t \mid t \bmod n \) (where $n$ is an integer and $x$ is a variable).

A \emph{(fixed-point-free) formula} is defined by the following grammar:
\begin{equation*}
    \varphi,\psi \quad::=\quad p(\vec{t}) \mid \neg p(\vec{t}) \mid P(\vec{t}) \mid \varphi \wedge \psi \mid \varphi \vee \psi \mid \forall x. \varphi \mid \exists x. \varphi,
\end{equation*}
where \( t \) is a term in \textsf{LIA}, \( p \) is an atomic predicate in \textsf{LIA} (\ie~\( = \) or \( \le \)), and \( P \) is a user-defined predicate.
For notational simplicity, we will also use logical formulas that do not follow the above syntax in the strict sense but can be transformed into the above form (\eg~\( (x < 42 \wedge i \neq 0) \Rightarrow P(x+i) \), which is equivalent to \( (42 \le x \vee i = 0 \vee P(x+1)) \)).
Predicates \( P \) are defined by mutual recursion:
\begin{equation*}
    P_1(\vec{x}_1) \stackrel{\Mode_1}{=} \varphi_1
    \:;\:
    P_2(\vec{x}_2) \stackrel{\Mode_2}{=} \varphi_2
    \:;\dots;\:
    P_n(\vec{x}_n) \stackrel{\Mode_n}{=} \varphi_n
\end{equation*}
where \( \vec{x}_i \) is a sequence of variables, \( \Mode_i \in \{ \nu, \mu \} \), and \( \varphi_i \) is a fixed-point-free formula with free variables in \( \{ \vec{x}_i, P_1,\dots,P_n \} \).
If \( \Mode_i = \nu \) (resp.~\( \mu \)), then \( P_i \) is the greatest solution (resp.~least solution) of the equation \( P_i(\vec{x}_i) = \varphi_i \) (for details, see the formal semantics defined below).
The order of definitions matters: \( \cdots;\: P_i(\vec{x}_i) \stackrel{\Mode_i}{=} \varphi_i \:;\: P_{i+1}(\vec{x}_{i+1}) \stackrel{\Mode_{i+1}}{=} \varphi_{i+1} \:; \cdots \) differs from \( \cdots;\: P_{i+1}(\vec{x}_{i+1}) \stackrel{\Mode_{i+1}}{=} \varphi_{i+1} \:;\: P_i(\vec{x}_i) \stackrel{\Mode_i}{=} \varphi_i \:; \cdots \) when \( \Mode_i \neq \Mode_{i+1} \).
For readers familiar with parity conditions, the left has higher priority.
A formula is a pair \( (\varphi, \Defs) \) of a fixed-point free formula \( \varphi \) with a mutually-recursive definition \( \Defs \) of predicates in \( \varphi \).

\begin{example}[{\citet[Example~1.1]{Heim2024}}]
    This example is a formalization of the following control problem.
    Consider a system with a single integer variable \( x \colon \Int \).
    For each step, the environment chooses \( i \in \Int \).
    If \( x < 42 \) or \( i = 0 \), the system successfully terminates.
    If \( x \ge 42 \) and \( i \neq 0 \), the system proceeds to the next step, updating \( x \).
    In this case, the controller can choose the next value of \( x \) from \( (x+i) \) and \( (x-i) \).
    The problem is to synthesize a controller that will always make the system terminate (or, if such a controller does not exist, to answer that it does not exist).

    The existence of a controller is equivalent to the validity of \( \forall x. P(x) \) where
    \begin{align*}
        P(x) &\quad\stackrel{\mu}{=}\quad \forall i. \Big(\big((x < 42 \vee i = 0\big) \Rightarrow \top) \wedge \big((x \ge 42 \wedge i \neq 0) \Rightarrow (P(x+i) \vee P(x-i))\big) \Big)
    \end{align*}
    The user-defined predicate \( P(n) \) means that the system with an appropriate controller terminates from the state \( x = n \).
    The right-hand-side of the definition describes the one-step transition of the system; the universal quantification \( \forall i \) expresses that \( i \) is chosen by the environment.
    The predicate \( P \) is defined as the least fixed-point; this reflects the requirement that the system should not perform infinite step transitions.
    The formula \( \forall x. P(x) \) is valid, as we shall see.
    \qed
\end{example}

We formally define the semantics of a mutually-recursive definition of predicates.
First, we deal with a single equation \( P(\vec{x}) \stackrel{\xi}{=} \varphi \).
We write \( \varphi(P)(\vec{x}) \) to make explicit that \( \varphi \) depends on \( P \) and \( \vec{x} \).
Let \( \Domain := (\Int^{\ell} \to \{0,1\}) \) be the semantic domain for the predicate \( P \), where \( \ell \) is the length of \( \vec{x} \).
This is a complete lattice by the point-wise order (\ie, for \( f,g \in \Domain \), \( f \le g \) if and only if \( f(\vec{n}) \le g(\vec{n}) \) for every \( \vec{n} \in \Int^{\ell} \)).
Since \( \neg P \) does not appear in \( \varphi \), the function \( \varphi \colon \Domain \longrightarrow \Domain \) mapping \( d \in \Domain \) to \( \varphi(d)\) (which is a predicate \( \Int^\ell \ni \vec{n} \mapsto \varphi(d)(\vec{n}) \in \{0,1\} \)) is a monotone function.
Then, Knaster-Tarski theorem shows that the set \( \{ d \in \Domain \mid d = \varphi(d) \} \) has both the least and greatest elements.
When \( \xi = \mu \) (resp.~\( \xi = \nu \)), then \( P \) is defined as the least solution (resp.~greatest solution).

The semantics of a mutually-recursive definition
\begin{equation*}
    \Defs
    \quad=\quad
    \big(
        P_1(\vec{x}_1) \stackrel{\Mode_1}{=} \varphi_1
        \:;\:
        P_2(\vec{x}_2) \stackrel{\Mode_2}{=} \varphi_2
        \:;\dots;\:
        P_n(\vec{x}_n) \stackrel{\Mode_n}{=} \varphi_n
    \big)
\end{equation*}
is given as follows.
Let \( \Domain_i := (\Int^{\ell_i} \to \{0,1\}) \) be the semantic domain for the predicate \( P_i \), where \( \ell_i \) is the length of \( \vec{x}_i \).
We solve the equations in the right-to-left direction.
A subtlety here is that \( \varphi_n \) may contain \( P_1,\dots,P_{n-1} \); we write \( \varphi_n(P_1,\dots,P_{n-1},P_n)(\vec{x}_n) \) to clarify the dependency.
For each \( (d_1,\dots,d_{n-1}) \in \Domain_1 \times \dots \times \Domain_{n-1} \),
the equation \( P_n(\vec{x}_n) \stackrel{\Mode_n}{=} \varphi_n(d_1,\dots,d_{n-1}, P_n)(\vec{x}_n) \) determines a value \( \lambda_n(d_1,\dots,d_{n-1}) \in \Domain_{n} \) depending on \( d_1,\dots,d_{n-1} \).
The solutions \( \lambda_n(d_1,\dots,d_{n-1}) \) parameterized by \( (d_1,\dots,d_{n-1}) \in \Domain_1 \times \dots \times \Domain_{n-1} \) determine a function \( \lambda_n \colon \Domain_1 \times \dots \times \Domain_{n-1} \to \Domain_n \).
Substituting \( P_n \) with \( \lambda_n(P_1,\dots,P_{n-1}) \)
yields a system of equations without \( P_n \), and then we solve the equation for \( P_{n-1} \).
Iteratively applying this process yields \( P_k = \lambda_k(P_1,\dots,P_{k-1}) \) for each \( k= 1,\dots,n \) depending on \( P_1,\dots,P_{k-1} \).
Let \( \delta_1 = \lambda_1() \), \( \delta_2 = \lambda_2(\delta_1) \), and \( \delta_k = \lambda_k(\delta_1,\dots,\delta_{k-1}) \) for general \( k \).
The assignment on \( P_1,\dots,P_n \) determined by \( \Defs \) is \( (\delta_1,\dots,\delta_n) \in \Domain_1 \times \dots \times \Domain_n \).
The formula \( (\varphi, \Defs) \) is \emph{valid} if \( \varphi = \varphi(\delta_1,\dots,\delta_n) \) is valid.

\subsection{Lagrangian and Primal-Dual Procedure}
As we have mentioned at the end of the previous subsection, the validity problem becomes easier if one can remove the least fixed-point \( \mu \) and existential quantifier \( \exists \).
We have already discussed a way to remove an existential quantifier \( \exists \) in \cref{sec:qlra}, in which we used a \textsf{SAT} strategy skeleton to guide a process of removal.
The removal of the least fixed-point has been addressed in \cref{sec:termination}, albeit in a slightly different form: a ranking function or disjunctively-well-founded relation is used to reduce the termination problem (a typical \( \mu \)-property) to a safety problem (a typical \( \nu \)-property).

The Lagrangian of this section simply combines these data, so the \( Y \)-component is a pair of a \textsf{SAT} strategy skeleton and a disjunctively-well-founded relation.
The \( X \)-component is similar: it describes choices of \( \forall \)-variables and ranking information for \( \nu \)-recursions.

We formally describe the details.
Assume a pair \( (\varphi, \Defs) \), where
\begin{equation*}
    \Defs
    \quad=\quad
    \big(
        P_1(\vec{x}_1) \stackrel{\Mode_1}{=} \varphi_1
        \:;\:
        P_2(\vec{x}_2) \stackrel{\Mode_2}{=} \varphi_2
        \:;\dots;\:
        P_n(\vec{x}_n) \stackrel{\Mode_n}{=} \varphi_n
    \big).
\end{equation*}
Let \( \ell_i \) be the length of \( \vec{x}_i \).
For simplicity, we assume that \( \varphi, \varphi_1,\dots, \varphi_n \) are in prenex normal form.
Let \( \RankingFunc_i \) be the set of candidate ranking functions available for \( P_i \).
Let \( L := \{ i \mid \xi_i = \mu \} \) and \( G := \{ i \mid \xi_i = \nu \} \) be the sets of indices of \( \mu \)- and \( \nu \)-predicates, repectively.

The Lagrangian is a function \( L_{\mathrm{Fix}} \colon X \times Y \longrightarrow \{-1,1\} \), where
\begin{align*}
    X &\textstyle \quad:= \prod_{i \in G} \finitepowerset(\RankingFunc_i) \times \UNSATSkeleton(\varphi) \times \prod_{i = 1}^n \UNSATSkeleton(\varphi_i)
    \\
    Y &\textstyle \quad:= \prod_{i \in L} \finitepowerset(\RankingFunc_i) \times \SATSkeleton(\varphi) \times \prod_{i = 1}^n \SATSkeleton(\varphi_i).
\end{align*}
Recall that in \cref{sec:qlra}, which dealt with logical expressions with quantifiers, the two arguments of the Lagrangian produce a logical formula whose validity is trivially computable (\ie~a logical formula without quantifiers).
A similar situation can be observed for fixed-point logic: two arguments of the Lagrangian induce a logic formula whose validity is computable.
The problematic constructs in fixed-point logic formulas are quantifiers and recursions;
the quantifiers are eliminated by the strategy skeletons, and we avoid the divergence in the evaluation of the recursion by immediately stopping the evaluation when a ranking information violation is detected.
Then, we define the value of the Lagrangian as the outcome of the evaluation.

Given \( x = ((R_i)_{i \in G}, \varrho, (\varrho_i)_{i = 1,\dots,n}) \in X \) and \( y = ((R_i)_{i \in L}, \pi, (\pi_i)_{i = 1,\dots,n}) \), we formally define the value \( L_{\mathrm{Fix}}(x,y) \) of the Lagrangian as follows.
We first apply the strategy skeletons \( (\varrho, (\varrho_i)_i) \) and \( (\pi, (\pi_i)_i) \) and construct a quantifier-free approximation \( (\varphi', \Defs') \) of the input \( (\varphi, \Defs) \).
This is simply obtained by applying the translation \( \langle \varrho \,|\, \varphi \,|\, \pi \rangle \) defined in \cref{sec:qlra}.
That is, \( \varphi' := \langle \varrho \,|\, \varphi \,|\, \pi \rangle \) and \( P'_i(\vec{x}_i) \stackrel{\xi_i}{=} \langle \varrho_i \,|\, \varphi_i\,|\, \pi_i \rangle \).
Note that \( (\varphi', \Defs') \) is neither over- nor under-approximation.

The formula \( (\varphi', \Defs') \) has no quantifier nor free variable, which we regard as a ``program'' with recursively defined functions \( P'_1,\dots,P'_n \) with two kinds of non-deterministic branches, demonic \( \wedge \) and angelic \( \vee \).
We run the program, recording all visited states and monitoring them for violations of the disjunctively-well-founded relations \( R_i \).
So, a configuration is a pair \( (\psi, \vec{V}) \) of a formula \( \psi \) and a list of sets \( \vec{V} = (V_1,\dots,V_n) \) in which \( V_i \) consists of actual arguments \( \vec{n}_i \) such that \( P_i(\vec{n}_i) \) has been visited.
The reduction relation is defined by
\begin{align*}
    (\psi_1 \mathrel{\square} \psi_2, (V_1,\dots,V_n)) &\longrightarrow (\psi_i, (V_1,\dots,V_n)) & {\square} \in \{ \wedge, \vee \},\: i \in \{ 1,2 \} \\
    (P'_i(\vec{n}_i), (V_1,\dots,V_n)) &\longrightarrow (\varphi'_i[\vec{n}_i/\vec{V}_i], (V_1,\dots,V_i\cup\{(\vec{n}_i)\},\dots,V_n)) & \forall \vec{m} \in V_i.\: \vec{m} \succ_{R_i} \vec{n}_i \\
    (P'_i(\vec{n}_i), (V_1,\dots,V_n)) &\longrightarrow \mathtt{false} & \xi_i = \mu \wedge \exists \vec{m} \in V_i.\: \vec{m} \not\succ_{R_i} \vec{n}_i \\
    (P'_i(\vec{n}_i), (V_1,\dots,V_n)) &\longrightarrow \mathtt{true} & \xi_i = \nu \wedge \exists \vec{m} \in V_i.\: \vec{m} \not\succ_{R_i} \vec{n}_i
\end{align*}
(where \(\varphi'\) is the formula defined above, stemming from the definition of \(P'_i\))
and this reduction process always terminates because \( \succ_{R_i} \) is disjunctively well-founded.
A normal form is \( \mathtt{true} \), \( \mathtt{false} \) or \( (\mathit{atomic}, \vec{V }) \) for some closed atomic formula \( \mathit{atomic} \), so a normal form is associated to a truth value.
Then we can reverse the reduction backwards and assign a truth value to each configuration.
A configuration \( (\psi, \vec{V}) \) with \( \psi = \psi_1 \wedge \psi_2 \) or \( \psi_1 \vee \psi_2 \) has successors, and its truth value is the conjuction or the disjunction of those of the successors depending on \( \psi \).
We define \( L_{\mathrm{Fix}}(x,y) = 1 \) if and only if the value assigned to \( (\varphi', (\emptyset,\dots,\emptyset)) \) is true.
Note that \( L_{\mathrm{Fix}}(x,y) \) is computable (provided that every ranking function in \( \RankingFunc_i \) is computable).

We say that a subset \( \RankingFunc \subseteq (Z \to \Nat) \) of ranking functions on a set \( Z \) is \emph{fine} if, for every \( x,y \in Z \), there exists \( r \in \RankingFunc \) such that \( r(x) > r(z) \).
\begin{theorem}
    Assume that \( \RankingFunc_i \) is fine for every \( i \).
    Then \( \sup_y \inf_x L_{\mathrm{Fix}}(x,y) = 1 \) implies the truth of the input formula.
\end{theorem}
\begin{proof}[Proof Sketch]
    We prove the contraposition.  Assume that the input formula is false, and assume \( y \in Y \).
    It suffices to find \( x \in X \) such that \( \sup_y \inf_x L_{\mathrm{Fix}}(x,y) = -1 \).
    By the game semantics of the fixed-point logic, the opponent has a winning strategy \( \sigma \) for the game induced by the input formula.
    For every proponent strategy \( \sigma' \), the play determined by \( \sigma \) and \( \sigma' \) reaches a false atomic formula or diverges with an infinitely deep recursive calls to a \( \mu \)-predicate.
    The evaluation of such a recursive call eventually violates the disjunctively-well-founded relation given by \( x \).
    So the fact that \( \sigma' \) loses against \( \sigma \) is witnessed by a finite set of finite reduction sequences.
    This is particularly true when \( \sigma' \) chooses each assignment of \( \exists \) following the \textsf{SAT} strategy skeletons in \( x \).
    When the proponent follows the \textsf{SAT} strategy skeleton, the proponent always has finitely many options at each moment.
    So the fact that the proponent cannot win the game following the \textsf{SAT} strategy skeletons in \( x \) can be witnessed by a finitely branching tree with no infinite path, so it is a finite tree by K\"onig's lemma.
    Thanks to the finiteness, together with the fineness assumption on \( \RankingFunc \), we can give \( y \) that beats \( x \).
\end{proof}

Both of the \( X \)- and \( Y \)-components are join semilattices, and the Lagrangian \( L_{\mathrm{fix}} \) is anti-monotone on \( X \) and monotone on \( Y \).
So the monotonicity requirement is satisfied.
Our implementation uses a variant of the basic primal-dual procedure that is monotone on both \( \alpha \in X \) and \( \beta \in Y \).

Let us examine the primal and dual witness check problems.
An interesting observation for the termination analysis with disjunctively-well-founded relations is that the dual witness check is reducible to a safety problem, using a transition system with an additional component that records one of the already visited states.
One may expect that the same idea should apply to the problem of this section as well, but here is a subtlety.
The soundness of the reduction relies on the setting of \cref{sec:termination}, in which all non-deterministic branches in the transition system are demonic branching \( \wedge \) (\ie~whatever the choice, the system should terminate).
In the presence of angelic branches \( \vee \), the reduction is unsound in general, because the proponent can make a choice that depends on the value of the additional component, which is a kind of cheating as the additional component is not a part of the real state of the transition system.
One way to restore soundness is to ensure that choices are independent of additional components by making the dependencies explicit using (fresh) functional symbols.
This is the approach of our implementation.

\begin{example}
    Conisder the predicate \( P \) defined by \( P(x) \stackrel{\mu}{=} \exists z. P(z) \).
    Then \( P(x) \) is false for every \( x \).
    Applying the translation in \cref{sec:termination:dwf}, we obtain \( P'(x,y) \stackrel{\nu}{=} \exists z. (y \succ_R z) \wedge P'(z,y) \wedge P'(z,z) \), but this formula becomes true when we appropriately choose a disjunctively-well-founded relation \( R \).
    Note that \( P' \) is the greatest fixed-point, corresponding to the fact that \cref{sec:termination:dwf} translates the termination problem to a safety problem.
    Let \( r(x) := \max(x, 0) \) and \( r'(x) := \max(3-x,0) \).
    Our strategy is to choose \( z \) from \( 0,1 \) so that \( z \neq y \).
    Then \( y \succ_R z \) holds for every \( y \), and thus \( P'(x,y) \) is true for every \( x \) and \( y \).
    Note that the choice of \( z \) depends on the additional component \( y \).

    Our implementation reduces the problem to find an appropriate function \( f \) such that \( P''(x,y) \stackrel{\nu}{=} (y \succ_R f(x)) \wedge P'(f(x),y) \wedge P'(f(x),f(x)) \).
    The definition of \( P'' \) is obtained by replacing \( z \) with \( f(x) \), expressing that the choice of \( z \) should only depend on \( x \).
    This approach is sound, and there is no pair \( (f,R) \) that makes \( P''(x,y) \) true.
    \qed
\end{example}

\subsection{Propotype Implementation and Evaluation}

We implemented the new method proposed in this section as {\sc MuStrat}, a prototype validity checker for a fixpoint-logic over quantified linear integer and real arithmetic. Additionally, for evaluation, we conducted comparative experiments with {\sc MuVal}~\cite{Unno2023}, the existing state-of-the-art fixpoint logic validity checker.  We ran the tools on (1) the 335 (non-)termination benchmarks from termCOMP (C Integer category), (2) the 202 fixpoint logic benchmarks provided by the authors of {\sc MuVal}~\cite{Unno2023}, and (3) the 47 game-solving benchmarks from \cite{Heim2024}, with a timeout of 300 seconds in the StarExec environment. The results are summarized in scatter plots, respectively in Figures~\ref{fig:termCOMP}, \ref{fig:popl2023}, and \ref{fig:popl2024}. In (1) and (2), although the total number of solved problems is slightly lower than {\sc MuVal}, {\sc MuStrat} successfully solved 6 and 7 problems respectively that MuVal timed out on. Additionally, in problems where strategy synthesis is crucial, such as non-termination verification and branching-time properties verification, MuStrat was able to find witnesses faster than {\sc MuVal} in more cases. (Note that in Figure~\ref{fig:termCOMP}, the red dots represent non-terminating instances.) {\sc MuStrat} also has an engineering advantage as it uses the highly efficient {\sc Spacer} as the backend CHC solver, allowing it to quickly solve examples that require the synthesis of complex inductive invariants. On the other hand, {\sc MuVal}, through CEGIS iterations, can simultaneously synthesize the necessary ranking functions and inductive invariants, which may be inter-dependent, allowing it to solve problem instances requiring complex ranking functions that {\sc MuStrat} times out on. On the other hand, the dedicated game solver {\sc rpgsolve}, proposed in the paper that provided the benchmark set (3)~\cite{Heim2024}, successfully solved all but two problems. To achieve performance close to {\sc rpgsolve}, implementation tuning unrelated to the core algorithm (e.g., efficiently handling the DAG-shaped resolution proofs returned by Spacer without expanding them into trees) is necessary. This is future work for tool development and is out of scope for the present paper. Additionally, it should be noted that {\sc rpgsolve} targets a limited subset of fixpoint logic, making it inapplicable to the benchmark sets (1) and (2).

\begin{figure}[t]
    \begin{tabular}{ccc}
        \begin{minipage}[t]{0.3\hsize}
            \centering
            \includegraphics[scale=0.33]{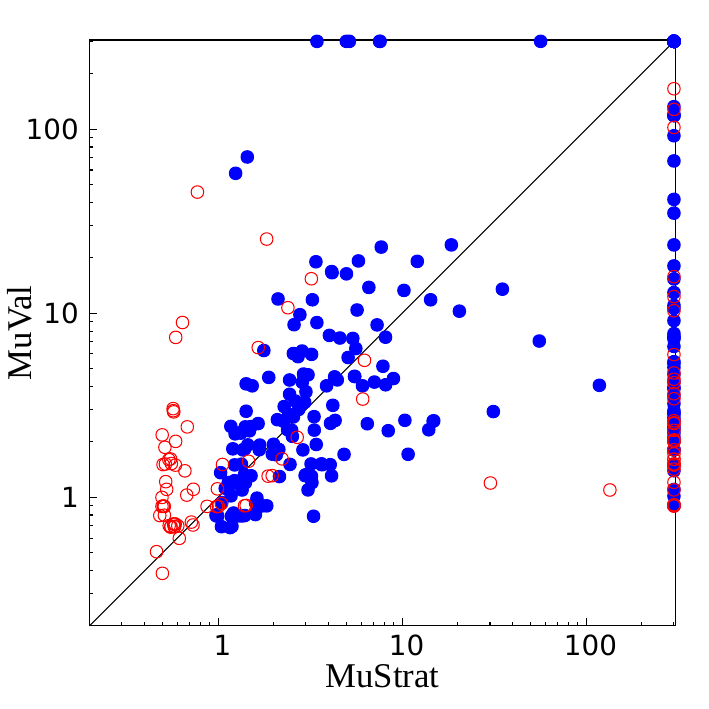}
            \caption{{\bf MuStrat} vs. {\bf MuVal} on the (non-)term. benchmark set from termCOMP (C Integer).}
            \label{fig:termCOMP}
        \end{minipage} &
        \begin{minipage}[t]{0.3\hsize}
            \centering
            \includegraphics[scale=0.33]{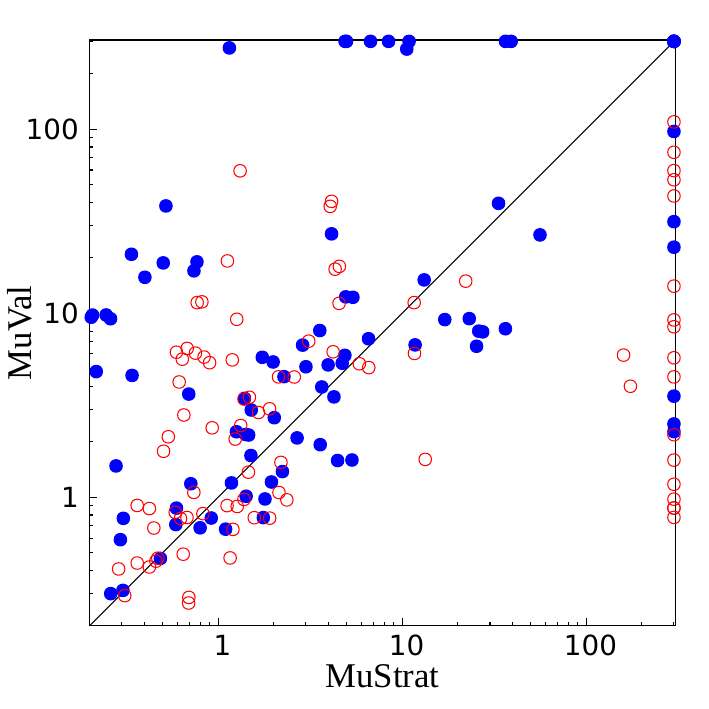}
            \caption{{\bf MuStrat} vs. {\bf MuVal} on the fixed-point logic benchmark set from \cite{Unno2023}.}
            \label{fig:popl2023}
        \end{minipage} &
        \begin{minipage}[t]{0.3\hsize}
            \centering
            \includegraphics[scale=0.33]{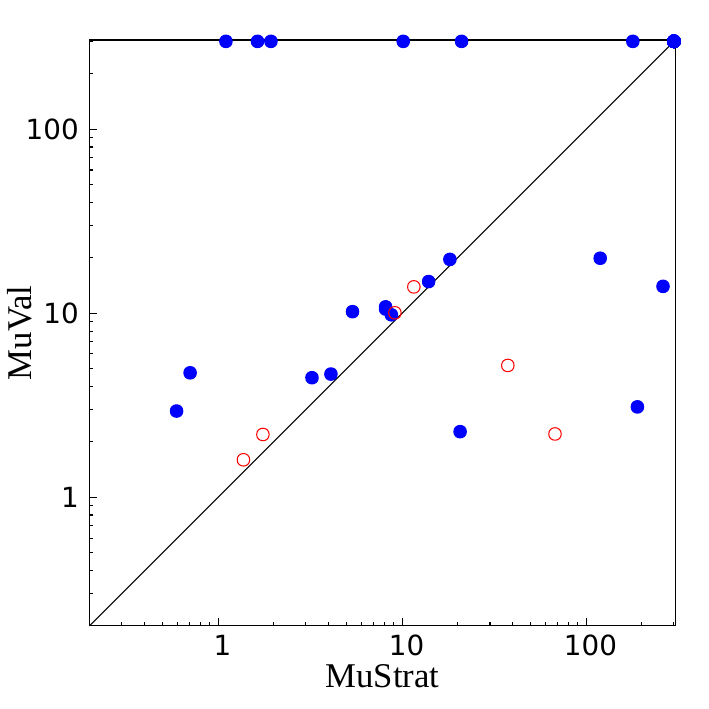}
            \caption{{\bf MuStrat} vs. {\bf MuVal} on the game solving benchmark set from \cite{Heim2024}.}
            \label{fig:popl2024}
        \end{minipage}
    \end{tabular}
\end{figure}

\section{Related Work and Concluding Remarks}\label{sec:related}

The development of unifying frameworks for algorithms in program analysis and verification is a common theme.
Abstract Interpretation~\cite{CC77,CC79} is perhaps the most prominent example of a mathematical framework that describes a variety of program analysis algorithms.
While in some instances the primal or dual problems in our framework can be seen as solved by computing an abstract fixpoint in a suitable domain (e.g., in CEGAR and in primal-dual Houdini), it seems to us that the entirety of \Cref{alg:lagrangian:procedure} is not an instance of abstract fixpoint computation.
We consider it an interesting avenue for future work to better understand the connection between our Lagrangian-based algorithmic framework and abstract interpretation. In particular, primal-dual search may be related to widening and narrowing.

In optimization and constraint solving, the so called inference dual and the relaxation dual (which is a special case of the inference dual) have been observed to underlie many algorithms~\cite{Hooker2000,Hooker2006}. In these terms, the duality of our generalized Lagrangian can be seen as both an inference dual and a relaxation dual, in the sense that, e.g., some $x \in X$ such that $\sup_y L(x,y) = v$ can be seen as a proof that $\inf_x \sup_y L(x,y) \leq v$.
In some cases \Cref{alg:lagrangian:procedure} can yield what \cite{Hooker2006} calls a nogood-based search method, but \Cref{alg:lagrangian:procedure} seems to be more general. Exploring the connection between our framework and more algorithms and mathematical frameworks for constraint satisfaction and optimization is an interesting avenue for future work.

Applications of linear/semidefinite programming and their (relaxation) duals to verification can be found in the literature.  \citet{Cousot2005} proposed an approach to reduce a verification problem into linear or semidefinite programming using Lagrangian relaxation as a key component.  \citet{Gonnord2015} synthesizes ranking functions by formulating the problem as a linear programming problem. Their use of extremal counterexamples to guarantee finiteness and therefore termination bears some resemblance to the choice of terms (mentioned at the end of \cref{sec:qlra}) that guarantee termination in the Quantified LRA solver discussed in \cref{sec:qlra}.

The duality in verification and constraint solving often appears as the duality between the proponent and opponent of a game.
For example, the quantified SMT solver by \citet{Farzan2016} studied in \cref{sec:qlra} is presented in terms of games.
Many quantified SMT solvers~\citet{Bjoerner2015,Murphy2024,DBLP:conf/cade/BonacinaGV23} are inspired by the game semantics of logical formulas and the duality inherent in it.
For example, \citet{Murphy2024} improved the approach by \citet{Farzan2016}: whereas the original method adds a strategy skeleton that could beat the current opponent strategy, the new method seeks a strategy skeleton that not only beats the current opponent strategy but also is winning on a subgame.
It is an interesting challenge to see if these advanced ideas can be understood and developed within a Lagrangian-based framework.

An important family of procedures that this paper does not deal with is IC3/PDR~\cite{Bradley2011,Een2011} and its relatives such as GPDR~\cite{Hoder2012} and \textsc{Spacer}~\cite{Komuravelli2014,Komuravelli2015,VediramanaKrishnan2023,Tsukada2024}.
We developed some candidate Lagrangians aiming to capture the behavior of PDR, some of which seem to capture PDR to a certain extent.
Detailed analysis of PDR based on Lagrangians is an important topic for future work.
According to \citet{Tsukada2022a}, a game solving procedure by \citet{DBLP:journals/pacmpl/FarzanK18} can be related to PDR, so the analysis of PDR would be beneficial to understand their procedure as well.

The primal-dual algorithms of~\cite{Padon2022,Unno2023} have inspired the development of our Lagrangian-based framework.
\Cref{sec:houdini} presented a suitable Lagrangian that shows that the high-level structure of primal-dual Houdini~\cite{Padon2022} is an instance of our framework. We believe the algorithm of~\cite{Unno2023} can similarly be captured by our framework, and we are still developing a suitable Lagrangian. One challenge there is that the primal and dual sides exchange information only under certain validity conditions. Therefore, either developing a suitable Lagrangian or possibly extending our framework is a planned future work.

\begin{acks}
We thank the anonymous reviewers for comments which improved the paper.
The research leading to these results has received funding from the
European Research Council under the European Union's Horizon 2020 research and
innovation programme (grant agreement No [759102-SVIS]).
This research was partially supported by the Israel Science Foundation (ISF) grant No.\ 2117/23 and \grantsponsor{JSPS}{JSPS}{} KAKENHI Grant Numbers \grantnum{JSPS}{JP20H04162}, \grantnum{JSPS}{JP20H05703}, \grantnum{JSPS}{JP19H04084}, \grantnum{JSPS}{JP24H00699}, \grantnum{JSPS}{JP23K24826}, and \grantnum{JSPS}{JP23K24820}.
This research was partially supported by a research grant from the Center for New Scientists at the Weizmann Institute of Science and
by a grant from the Azrieli Foundation.
\end{acks}

\bibliographystyle{ACM-Reference-Format}
\bibliography{jabref,manual}

\ifappendix
\clearpage
\appendix
\section{Proof of Theorem~\ref{thm:houdini:well-definedness}}\label{sec:appx:proof-houdini}

In this proof, we will refer to ``\( \TSys \) with bad states \( B \)'' the transition system obtained by replacing the bad states of \( \TSys \) with \( B \).
The key to the proof is the following claim:
\begin{claim*}
For each \( p \in \badState[\TSysI] \), at least one of the following holds:
\begin{enumerate}
    \item some \( \vartheta \subseteq \PredSet \) is an inductive invariant for \( \TSys \) with bad states \( \{ s \mid s \not\models p \} \), or
    \item some \( \varpi \subseteq \stateSet[\TSys] \) is an inductive invariant for \( \TSysI \) with bad states \( \{ p \} \).
\end{enumerate}
\end{claim*}
\begin{proof}
Let \( \vartheta \subseteq \PredSet \) be the maximum inductive invariant of \( \PredSet \).
Note that this is the ``strongest'' invariant since \( \vartheta \) is regarded as the conjunction \( \bigwedge_{p \in \vartheta} p \).
The maximum inductive invariant exists since, for a family \( \vartheta_i \subseteq \PredSet \) of inductive invariants, their union \( \bigcup_i \vartheta_i \) is also an inductive invariant.

Let \( \varpi \subseteq \stateSet[\TSys] \) be the subset of states satisfying all \( p \in \vartheta \), i.e.~\( \varpi := \{ s \in \stateSet[\TSys] \mid \forall p \in \vartheta.\: s \models p \} \).
Since \( \vartheta \) is an invariant, \( \varpi \) is closed under transition.
Then \( \varpi \models p \Longleftrightarrow p \in \vartheta \).
The right-to-left direction comes from the definition of \( \varpi \).
To see the other direction, assume \( \varpi \models p \).
Then \( \vartheta \cup \{ p \} \) is an invariant: if \( s \models \vartheta \cup \{ p \} \) and \( s \trans[\TSys] s' \), then \( s' \models \vartheta \) by the invariance of \( \vartheta \), so \( s' \in \varpi \) that implies \( s' \models p \).
By the maximality of \( \vartheta \), we have \( \vartheta = \vartheta \cup \{p\} \), that means, \( p \in \vartheta \).

The maximum inductive invariant \( \vartheta \) is closed under the transition.
Assume \( p \trans[\TSysI] p' \) and \( p \in \vartheta \).
It suffices to show that \( \vartheta \cup \{ p' \} \) is an inductive invariant (because \( \vartheta \) is the maximum invariant).
Assume \( s \trans[\TSys] s' \) and \( s \models \vartheta \cup \{ p' \} \).
Then \( s' \models \vartheta \) since \( \vartheta \) is an inductive invariant.
We have \( s \trans[\TSys] s' \), \( p \trans[\TSysI] p' \), \( s \models p \), \( s' \models p \) and \( s \models p' \).
By the duality, \( s' \models p' \).
So \( s' \models \vartheta \cup \{ p' \} \) as expected.

This means that \( \varpi \) is a dual-inductive invariant.
Assume a transition \( p \trans[\TSysI] p' \) and \( \varpi \models p \).
Then \( p \in \vartheta \).
Since \( \vartheta \) is closed under \( \rightsquigarrow \), we have \( p' \in \vartheta \).
So \( \varpi \models p' \).

Now we have a pair \( (\vartheta, \varpi) \) of an inductive invariant \( \vartheta \) and a dual-inductive invariant \( \varpi \) related by \( \varpi \models p \Leftrightarrow p \in \vartheta \).
If \( q \in \vartheta \), the first condition is met.
If \( q \notin \vartheta \), then \( \neg (\varpi \models q) \), i.e.~there exists \( s \in \varpi \) such that \( x \not\models p \).
So the dual-inductive invariant \( \varpi \) witnesses the safety of \( \TSysI \) with respect to \( \{ q \} \).
\end{proof}

We prove the theorem using the claim.

Assume that the latter case holds for every \( p \in \badState[\TSysI] \).
Let \( \varpi_p \subseteq \stateSet[\TSys] \) be a safe inductive invariant for \( \TSysI \) with bad states \( \{ p \} \).
Since the set of inductive invariants is closed under the union, \( \varpi := \bigcup_p \varpi_p \) is an inductive invariant for \( \TSysI \).
For every \( p \in \badState[\TSysI] \), we have \( p \not\models \varpi_p \) and hence \( p \not\models \varpi \).
So \( \varpi \) is a safe inductive invariant for \( \TSysI \).

Assume that the former holds for some \( p \in \badState[\TSysI] \).
That means, the safety of \( \TSys \) with respect to \( \{ s \mid s \not\models p \} \) is witnessed by some \( \vartheta \subseteq \PredSet \).
Then \( s \models \vartheta \Rightarrow \neg (s \not\models p) \Leftrightarrow s \models p \Rightarrow s \not\in \badState[\TSys] \).
So \( \vartheta \) witnesses the safety of \( \TSys \) with respect to \( \badState[\TSys] \).

\section{Proof of Lemma~\ref{lem:qlra:smt}}\label{sec:appx:proof-of-smt-correctness}
To prove the right-to-left direction, we prove the following claim by induction on \( \pi \): if \( L^\varphi_{\mathrm{FK}}(\varrho, \pi) = -1 \), then \( \varphi | \pi \rangle \) is invalid.
The claim is trivial for \( \pi = \bullet \).
Assume \( \pi = \forall x. \pi' \).  Then \( \varphi = \forall x. \varphi' \) and \( \varrho = \bigsqcap_{i = 1}^n t_i.\varrho_i \).
Then \( \langle \varrho | \varphi | \pi \rangle = \bigwedge_{i=1}^n \langle \varrho_i | \varphi'[t_i/x] | \pi'[t_i/x] \rangle \), which is invalid.
So \( \langle \varrho_i | \varphi'[t_i/x] | \pi'[t_i/x] \rangle \) is invalid for some \( i \).
By the induction hypothesis, \( \varphi'[t_i/x] | \pi'[t_i/x] \rangle \) is invalid.
It is not difficult to see that \( (\varphi'[t_i/x] | \pi'[t_i/x] \rangle) = (\varphi' | \pi' \rangle)[t_i/x] \), so \( \forall x. (\varphi' | \pi' \rangle) \) is invalid.
Assume \( \pi = t.\pi' \).
Then \( \varphi = \exists x. \varphi' \) and \( \varrho = \exists x. \varrho' \).
Since \( \langle \varrho | \varphi | \pi \rangle = \langle \varrho'[t/x] | \varphi'[t/x] | \pi' \rangle \), by the induction hypothesis, \( (\varphi'[t/x] | \pi' \rangle) = (\varphi | \pi' \rangle) \) is invalid.
Assume \( \pi = \pi_1 \sqcup \pi_2 \).
Then \( \langle \varrho | \varphi | \pi \rangle = \langle \varrho | \varphi | \pi_1 \rangle \vee \langle \varrho | \varphi | \pi_2 \rangle \) and this formula is invalid.
So both \( \langle \varrho | \varphi | \pi_1 \rangle \) and \( \langle \varrho | \varphi | \pi_2 \rangle \) are invalid, and by the induction hypothesis, both \( \varphi | \pi_1 \rangle \) and \( \varphi | \pi_2 \rangle \) are invalid.
So \( (\varphi | \pi \rangle) = (\varphi | \pi_1 \rangle) \vee (\varphi | \pi_2 \rangle) \) is invalid.

To prove the left-to-right direction, assume that \( \varphi|\pi\rangle \) is invalid.
We construct an \textsf{UNSAT} strategy skeleton \( \varrho \) such that \( L^\varphi_{\mathrm{FK}}(\varrho, \pi) = -1 \) by induction on \( \pi \).
If \( \pi = \bullet \), let \( \varphi = \bullet \).
Assume \( \pi = \forall x. \pi' \).  Then \( \varphi = \forall x. \varphi' \) and \( (\varphi|\pi\rangle) = \forall x. (\varphi'|\pi'\rangle) \).
Since this formula is invalid, there exists a rational number \( r \in \Rational \) such that \( (\varphi'|\pi'\rangle)[r/x] = (\varphi'[r/x]|\pi'[r/x]\rangle) \) is invalid.
By the induction hypothesis, there exists \( \varrho' \) such that \( L^{(\varphi'[r/x])}_{\mathrm{FK}}(\varrho', \pi') = -1 \).
Then \( \varrho := r.\varrho' \) satisfies the requirement.
Assume \( \pi = t.\pi' \).  Then \( \varphi = \exists x. \varphi' \) and \( (\varphi|\pi\rangle) = (\varphi'[t/x]|\pi'\rangle) \), which is invalid.
By the induction hypothesis, there exists \( \varrho' \) such that \( L^{(\varphi'[t/x])}_{\mathrm{FK}}(\varrho', \pi') = -1 \).
Then \( \varrho := \exists x. \varrho' \) satisfies the requirement.
Assume \( \pi = \pi_1 \sqcup \pi_2 \).  Then \( (\varphi|\pi\rangle) = (\varphi|\pi_1\rangle) \vee (\varphi|\pi_2\rangle) \), which is invalid.
By the induction hypothesis, we have \( \varrho'_1 \) and \( \varrho'_2 \) such that both \( \langle \varrho'_1|\varphi|\pi_1\rangle \) and \( \langle \varrho'_2|\varphi|\pi_2\rangle \) are invalid.
Then \( \varrho := \varrho_1 \cap \varrho_2 \) satisfies the requirement.

\fi

\end{document}